\documentclass[10pt,article,letterpaper]{IEEEtran}
\usepackage[utf8]{inputenc}

\onecolumn
\title{Feedback Capacity of MIMO Gaussian Channels}
\author{Oron Sabag, Victoria Kostina, Babak Hassibi}
\date{November 2020}
\usepackage{amsmath, amssymb,amsthm,cite}
\usepackage{graphicx}
\usepackage{psfrag}
\usepackage{bbm}
\usepackage{mathtools}
\usepackage{color}
\usepackage{tcolorbox}
\usepackage{algorithm,algpseudocode}

\DeclareMathOperator*{\nn}{\nonumber}
\DeclareMathOperator{\E}{\mathbb{E}}
\DeclareMathOperator{\vy}{\mathbf{y}}
\DeclareMathOperator{\ve}{\mathbf{e}}
\DeclareMathOperator{\vm}{\mathbf{m}}
\DeclareMathOperator{\vv}{\mathbf{v}}
\DeclareMathOperator{\vw}{\mathbf{w}}
\DeclareMathOperator{\cov}{\mathbf{cov}}
\DeclareMathOperator{\vz}{\mathbf{z}}
\DeclareMathOperator{\vx}{\mathbf{x}}
\DeclareMathOperator{\vs}{\mathbf{s}}
\DeclareMathOperator{\vhs}{\hat{\mathbf{s}}}
\DeclareMathOperator{\vhhs}{\hat{\hat{\mathbf{s}}}}

\allowdisplaybreaks
\newtheorem{lemma}{Lemma}

\newtheorem{theorem}{Theorem}

\newtheorem{remark}{Remark}
\theoremstyle{definition}

\newtheorem{assumption}{Assumption}

\newtheorem{definition}{Definition}

\newcommand\rr[1]{\textcolor{black}{#1}}

\makeatletter
\def\blfootnote{\gdef\@thefnmark{}\@footnotetext}
\makeatother

\begin{document}

\maketitle

\begin{abstract}
Finding a computable expression for the feedback capacity of channels with colored Gaussian, additive noise is a long standing open problem. In this paper, we solve this problem in the scenario where the channel has multiple inputs and multiple outputs (MIMO) and the noise process is generated as the output of a time-invariant state-space model. Our main result is a computable expression for the feedback capacity in terms of a finite-dimensional convex optimization. The solution to the feedback capacity problem is obtained by formulating the finite-block counterpart of the capacity problem as a \emph{sequential convex optimization problem} which leads in turn to a single-letter upper bound. This converse derivation integrates tools and ideas from information theory, control, filtering and convex optimization. A tight lower bound is realized by optimizing over a family of time-invariant policies thus showing that time-invariant inputs are optimal even when the noise process may not be stationary. The optimal time-invariant policy is used to construct a capacity-achieving and simple coding scheme for scalar channels, and its analysis reveals an interesting relation between a smoothing problem and the feedback capacity expression.
\end{abstract}






\section{Introduction}
\blfootnote{The authors are with the
Department of Electrical Engineering at California Institute of Technology (e-mails:
 \{oron,vkostina,hassibi\}@caltech.edu). Part of this work was published in \cite{sabag_MIMO_isit}.}

We consider the feedback capacity of a multiple-input multiple-output (MIMO) Gaussian channel
\begin{align}\label{eq:intro_channel}
    \vy_i&= \Lambda\vx_i + \vz_i,
\end{align}
where $\Lambda\in\mathbb{R}^{p\times m}$ is a deterministic matrix, $\vy_i$ is the channel output and $\vx_i$ is the channel input. The noise is a colored Gaussian process generated by a vector state-space model (a hidden Markov model)
\begin{align}\label{eq:intro_noise}
\vs_{i+1}&= F \vs_i + G \vw_i \nn\\
\vz_i&= H \vs_i + \vv_i,
\end{align}
where the sequence $(\vw_i,\vv_i)$ is i.i.d. with Gaussian distribution. Our assumptions on the state-space are mild and include for instance non-stationary noise processes (when the spectral radius of $F$ is greater than $1$). Particular realizations of the state-space reveal known random processes such as the auto-regressive moving-average (ARMA) noise process.


Most related to our setting is the framework for channels with general additive Gaussian noise processes by Cover and Pombra \cite{CoverPombra}. They showed that the feedback capacity is equal to the limit of
\begin{align}\label{eq:intro_nletter}
    C_n(P)&= \max_{K_V\succeq0,B} \frac1{2n} \log \frac{\det (K_V + (I+B)K_n^Z (I+B)^T)}{\det K_n^Z},
\end{align}
where $K_n^Z$ is the covariance of the Gaussian noise, and the maximum is subject to strictly-causal linear operators $B$ (lower-triangular matrices) and pairs $(K_V,B)$ that satisfy the power constraint $\text{Tr} (K_V + B K_n^ZB^T)\le n P$. Their general methodology applies to arbitrary  Gaussian processes, and can be extended to MIMO channels and results in a formula that is similar to \eqref{eq:intro_nletter}, but the computation of such expressions remains non-trivial. In this paper, we show that imposing a state-space structure on the Gaussian noise leads to a computable characterization of the infinite-limit of the optimization problem \eqref{eq:intro_nletter}.


Our main result is a computable expression for the feedback capacity, formulated as a finite-dimensional convex optimization problem. The optimization is a  maximal determinant optimization problem subject to linear matrix inequalities (LMIs) constraints, a class of convex optimization problems that often appear in the control literature \cite{boyd_detmax,boyd1994linear,scherer2000linear,caverly2019lmi} and recently also in information theory \cite{sabag_CDC_SDP,Tanaka_SDP_1,Tanaka_SDP_2,Gattami}. The LMIs are interpretable, and one of the LMIs corresponds to a tight relaxation of a Riccati equation. Several aspects of the feedback capacity solution such as computability, comparison with non-feedback rates, and optimal inputs distribution are discusses by studying the capacity of the moving-average (MA) and the auto-regressive (AR) noise processes.

The literature on the feedback capacity of scalar Gaussian channel is rich, e.g. \cite{Butman69,Butman_conjecture,TiernanSchalk_AR_UB,Ebert,Feder_Gaussian,YoulaCoding,Han_GaussianFeedback,KalmanConnectionISIT20,elia_bodemeets,Ozarow90}, and the focus here is on works most related to ours (a detailed survey can be found in \cite{Kim10_Feedback_capacity_stationary_Gaussian}). In \cite{YangKavcicTatikondaGaussian}, an explicit lower bound for ARMA(1,1) noise was derived. Their lower bound was shown to be optimal for the MA(1) noise in \cite{Kim06_MA}, and the conjecture was proven for the ARMA(1,1) in \cite{Kim10_Feedback_capacity_stationary_Gaussian}. \textcolor{black}{The capacity characterization for the ARMA(1,1) noise in \cite{Kim10_Feedback_capacity_stationary_Gaussian} relies on their general result that stationary channel input processes achieve the feedback capacity when the noise is stationary. Based on this fundamental result, \cite{Kim10_Feedback_capacity_stationary_Gaussian} also studied the special case of our channel in \eqref{eq:intro_channel}-\eqref{eq:intro_noise} where the channel is scalar, the noise is stable, and the hidden state is available to the encoder ($\vw_i=\vv_i$), and formulated its capacity as a finite-dimensional, non-convex optimization problem. In contrast, we provide a convex optimization for the feedback capacity of the general channel in \eqref{eq:intro_channel}-\eqref{eq:intro_noise} under mild conditions (see Section \ref{sec:setting} below). In \cite{Gattami}, a change of variable to the non-convex optimization in \cite{Kim10_Feedback_capacity_stationary_Gaussian}, combined with the novel idea of using LMIs, showed that the capacity can be formulated as a convex optimization problem. However, the change of variable relied on an erroneous claim (see Remark~\ref{remarkgattami}). Our paper studies colored Gaussian noise described by the general state-space model where the hidden state of the noise may or may not be available to the encoder, the channel may be scalar or vector (MIMO), and the noise may be stationary or not. We express the feedback capacity as a convex optimization problem in this general setting. A recent conference paper also studies MIMO channels, and extends the convex optimization in [15] to MIMO channels with ISI [24]. The intersection of \cite{Elia_MIMO_ITW} and our setting is a MIMO channel with a stable, colored Gaussian noise, and the capacity results in the current paper (initially published as \cite{sabag_MIMO_isit}) and \cite{Elia_MIMO_ITW} were developed independently and published concurrently. Each work considers a different extension of the MIMO channel with stable noise: \cite{Elia_MIMO_ITW} studies channels with ISI, whereas we study colored Gaussian noise that can be either stationary or non-stationary. A major technical contribution in our work is that we show the optimality of stationary inputs for non-stationary noise processes. This fact and its proof may be of independent interest since it does not rely on the frequency-based methods of \cite{Kim10_Feedback_capacity_stationary_Gaussian}.}

The starting point of our derivations is the general Cover-Pombra characterization in \eqref{eq:intro_nletter}, and we develop a \emph{time-domain} methodology to provide a computable expression for the feedback capacity. We derive a novel formulation of the $n-$letter capacity in \eqref{eq:intro_nletter} as a \emph{sequential convex optimization problem} (SCOP). In particular, we formulate an optimization problem whose decision variable is a sequence of length $n$, where at each time fixed-dimensional matrices are optimized. In the SCOP formulation, the LMI constraints have a sequential nature and should depend on two consecutive times only. This sequential property combined with the convexity of the problem is the key to obtain a single-letter upper bound for the limit of the $n$-letter capacity in \eqref{eq:intro_nletter}. For the lower bound, a family of time-invariant channel inputs distributions is optimized and is shown to achieve the upper bound. An outcome of our derivation is a new methodology to show that time-invariant inputs are sufficient to achieve the feedback capacity even when noise may be non-stationary.

An optimal time-invariant policy can be computed directly from the feedback capacity convex optimization. Using this policy, we also construct an explicit coding scheme that achieves the feedback capacity for scalar channels. The derived scheme generalizes the coding proposed in \cite{Kim10_Feedback_capacity_stationary_Gaussian}, and simplifies its encoding by showing that the message can be encoded in a single dimension rather than the multi-dimensional variant proposed in \cite{Kim10_Feedback_capacity_stationary_Gaussian}. We also derive an explicit decoding rule by studying a related smoothing problem \cite{kailath_booklinear}. The analysis of the smoothing problem reveals an interesting relation between the volume reduction of its error covariance and the capacity solution. That analysis is performed for the general case of a MIMO channel, and a possible MIMO scheme is discussed in Section \ref{sec:conclusion}.

The rest of the paper is organized as follows. In Section \ref{sec:setting}, we present the setting and the preliminaries. Section \ref{sec:main} includes our main result on the feedback capacity of the MIMO Gaussian channel and several examples. In Section \ref{sec:scheme}, we present the optimal inputs distribution and the capacity-achieving coding scheme. In Section \ref{sec:derivation}, the main ideas and the technical lemmas to prove our main result are presented while their detailed proofs appear in Section \ref{sec:proofs}.



\section{The setting and Preliminaries}\label{sec:setting}
This section includes the communication setting. We also present the Kalman filter and the Riccati equation that are required for the presentation of the main result.
\subsection{The setting}
We consider a MIMO additive Gaussian channel
\begin{align}\label{eq:pre_channel}
    \vy_i&= \Lambda \vx_i + \vz_i,
\end{align}
where the channel input is $\vx_i\in\mathbb{R}^{m}$, $\vy_i\in\mathbb{R}^p$ is the channel output, the additive noise is $\vz_i\in\mathbb{R}^p$, and $\Lambda\in\mathbb{R}^{p\times m}$ is a fixed known matrix. The encoder has access to noiseless, strictly-causal feedback so that the input $\vx_i$ is a function of the message and all previous channel outputs $\vy^{t-1}:= \vy_1,\dots,\vy_{t-1}$. For a fixed blocklength $n$, the channel input should satisfy the average power constraint $\frac1{n} \sum_{i=1}^n \E[\vx_i^T\vx_i]\le P.$ Definitions of the average probability of error, achievable rates and the feedback capacity are standard and can be found in \cite{Kim10_Feedback_capacity_stationary_Gaussian}, for instance. The feedback capacity with a power constraint $P$ is denoted by $C_{fb}(P)$.

In the case of MIMO channels, the capacity can be expressed as the multi-letter expression in \eqref{eq:intro_nletter} by modifying all the matrices to their corresponding block matrices with appropriate dimensions. An equivalent characterization of the $n$-letter objective in \eqref{eq:intro_nletter} is the directed information $I(\vx^n\to\vy^n)$ that characterizes the feedback capacity of point to point channels \cite{Kramer98,TatikondaMitter_IT09,PermuterWeissmanGoldsmith09,Massey90}.

The additive noise is a colored Gaussian process generated as the output of the state-space:
\begin{align}\label{eq:noise_state_space}
\vs_{i+1}&= F\vs_i + G \vw_i\nn\\
\vz_i&= H\vs_i + \vv_i,
\end{align}
where \textcolor{black}{$\vs_i\in\mathbb{R}^n$}, $\vw_i\sim N(0,W)$ and $\vv_i\sim N(0,V)$ are i.i.d. sequences with $\mathbb{E}[\vw_i\vv_i^T]=L$, and are independent of the initial state $\vs_1\sim N(0,\Sigma_{1})$. Note that the encoder has a strictly causal access to the noise $\vz_i$, but not to the hidden state $\vs_i$. For this case, we can use Kalman filtering in order to write the state-space in an observer form \cite{charalambous2020new}. This pre-Kalman filtering step for the state-space \eqref{eq:noise_state_space} allows one to define a new channel state that is available to the encoder as presented in the next section.



\subsection{The Kalman filter and the innovations process}
The Kalman filter is a simple, recursive method to compute the maximum likelihood estimate of the hidden state $\vs_i$ based on the measurements $\vz_1,\dots,\vz_{i-1}$. The predicted-estimate of the state and its error covariance are defined as
\begin{align}\label{eq:KF_estimator}
    \vhs_i &= \E[\vs_i|\vz^{i-1}]\nn\\
    \Sigma_i &= \cov(\vs_i-\vhs_i).
\end{align}
The Kalman filter is given by the recursion
\begin{align}\label{eq:prel_estimator}
    \vhs_{i+1}&= F\vhs_{i} + K_{p,i} (\vz_i-H\vhs_{i})
\end{align}
with the initialization $\vhs_1=0$, and the constants are
\begin{align}\label{eq:def_Kp_Psi}
K_{p,i}&= (F \Sigma_iH^T + GL)\Psi_i^{-1}\nn\\
\Psi_i&=  H \Sigma_i H^T+V,
\end{align}
where the error covariance $\Sigma_i$ is described by the Riccati recursion
\begin{align}\label{eq:Riccati_recursion}
\Sigma_{i+1}&= F \Sigma_{i}F^T + GWG^T - K_{p,i}\Psi_iK_{p,i}^T.
\end{align}
with the initial condition $\Sigma_1\succeq0$. The estimate in \eqref{eq:prel_estimator} can be computed using the \emph{innovations process} defined as $\ve_i = \vz_i-H\vhs_{i}$ and is distributed according to $N(0,\Psi_i)$. It is also known that the innovation $\ve_i$ is orthogonal (statistically independent) to the previous instances of the measurements $\vz^{i-1}$ \cite{Kailath_innovations}. \textcolor{black}{Thus, we can write a new equivalent channel as
\begin{align}\label{eq:new_channel}
        \vhs_{i+1}&= F\vhs_{i} + K_{p,i} \ve_i\nn\\
        \vy_i&= H\vhs_i + \Lambda \vx_i + \ve_i
\end{align}
where $\vhs_i$ plays the role of the channel state and is available to the encoder. Note that this is a valid channel due to the Markov chain $(\vhs_{i+1},\vy_i) - (\vx_i,\vhs_i) - (\vx^{i-1},\vy^{i-1},\vhs^{i-1},m)$.}

The innovations process also characterizes the entropy rate of Gaussian random processes as
\begin{align}\label{eq:z_entropy_rate}
    \frac1{n} h(\vz^n)&= \frac1{n} \sum_{i=1}^n h(\vz_i|\vz^{i-1})\nn\\
    &= \frac1{n} \sum_{i=1}^n h(\ve_i)\nn\\
    &= \textcolor{black}{\frac1{2n} \sum_{i=1}^n \log \det (\Psi_i) + \frac1{2n}\log (2\pi e)^d
     .}
\end{align}

In \eqref{eq:def_Kp_Psi}, it is assumed that $\Psi_i\succ0$ for all $i$. This is a natural assumption since otherwise the capacity is infinite. Namely, if $\Psi_i$ is only positive semidefinite, a coordinate in the noise vector $\vz_i$ is a deterministic function of the past noise instances $\vz^{i-1}$. Building an infinite-rate scheme is straightforward: the encoder transmits $\vx_j=0$ so that $\vy_j=\vz_j$ for $j\le i-1$. Then, by having $\vz^{i-1}$, the encoder and the decoder know a coordinate of $\vz_i$, and can communicate an inifinite number of bits on this vector coordinate (assuming the image of $\Lambda$ is not degenerated at this particular direction).

\subsection{The Riccati equation}\label{subsec:Riccati}
Consider the function
\begin{align}\label{eq:riccati_general}
    f(\Sigma)&= F \Sigma F^T - \Sigma + GWG^T - K_p(\Sigma)\Psi(\Sigma) K_p^T(\Sigma),
\end{align}
where $K_p(\Sigma) = (F\Sigma H^T + GL) \Psi(\Sigma)^{-1}$ and $\Psi(\Sigma) = H\Sigma H^T + V$. The Riccati equation is defined as $f(\Sigma)=0$. The stabilizing solution to the Riccati equation (if exists) not only solves $f(\Sigma)=0$, but is also the unique solution such that its corresponding closed-loop system $F - K_p(\Sigma) H$ is stable. In the rest of the paper, we refer to
\begin{align}\label{eq:ricc_constants}
    K_p &= (F\Sigma H^T + GL) \Psi^{-1}\nn\\
    \Psi &= H\Sigma H^T + V
\end{align}
as the constants evaluated at the stabilizing solution. \textcolor{black}{The corresponding time-invariant Kalman filter is
\begin{align}\label{eq:enc_KF_LTI}
    \vhs_{i+1}&= F\vhs_i + K_p(\vz_i-H\vhs_i).
\end{align}}

We move on to present assumptions on the state-space model. The stability of $F$ determines the stationarity of the noise process.
\begin{definition}\label{def:stable}
The matrix $F$ is stable if its spectral radius satisfies $\rho(F)<1$.
\end{definition}
Without further assumptions, our results hold for the stationary case, i.e., when $F$ is stable \textcolor{black}{(and $L=0)$}. Thus, a reader whose interest is limited to the stationary case may skip the following assumptions.
\begin{assumption}
The pair $(F,H)$ is detectable. That is, there exists a matrix $K$ such that $\rho(F-KH)<1$.
\end{assumption}
\begin{assumption}
The pair $(F_s,\textcolor{black}{W_s})$ is controllable on the unit circle where $F_s\triangleq F - GLV^{-1}H$ \textcolor{black}{and $W_s\triangleq W - GLV^{-1}L^TG^T$}. That is, for any $x$ and $\lambda$ such that $xF_s = x\lambda$, if $|\lambda|= 1$, then \textcolor{black}{$xW_s^{1/2}\neq0$}.
\end{assumption}
Assumption $1$ asserts that all eigenvectors of $F$ that have unstable eigenvalues (outside the unit circle) can be observed via the matrix $H$. Indeed, without loss of generality, it can even be assumed the pair $(F,H)$ is observable (for all eigenvectors) since the unobserved eigenvectors have no effect on the channel noise. Assumptions $1$ and $2$ are sufficient and necessary conditions for the existence of the unique stabilzing solution to the Riccati equation in \eqref{eq:riccati_general}.

\textcolor{black}{We further need to assume that the initial covariance matrix $\Sigma_1$ converges to the stabilizing solution. Advanced discussions on convergence of Riccati recursions can be found in \cite[App. E]{kailath_booklinear}, and here we aim to provide several alternatives in order to obtain a general framework. The first condition is stabilizability of $(F_s,W_s)$. That is, for any $x$ and $\lambda$ such that $xF_s = x\lambda$, if $|\lambda|\ge 1$, then $xW_s^{1/2}\neq0$. This condition guarantees that the stabilizing solution is the only positive semidefinite solution to the Riccati equation, which implies that \emph{any initial state covariance} converges to the stabilizing solution. Another useful condition is $\Sigma_1\succeq\overline{\Sigma}$ where $\overline{\Sigma}$ is the stabilizing solution. Beyond these two sufficient conditions, in simple cases (such as the moving average noise in Section \ref{subsec:MA}), the convergence can be verified manually. We can assume without loss of generality that the initial covariance $\Sigma_1$ is equal to the stabilizing solution of the Riccati equation in \eqref{eq:riccati_general} by letting $\vx_i=0$ before the transmission begins.}



\section{Main result and discussion}\label{sec:main}
In this section we present the feedback capacity of the MIMO channel and its particularization to scalar channels. We discuss different aspects of our main results via several examples. The following is our main result.
\begin{theorem}[Feedback capacity of MIMO channels]\label{th:main}
The feedback capacity of the MIMO Gaussian channel in \eqref{eq:pre_channel}-\eqref{eq:noise_state_space} is given by the
convex optimization
\begin{align}\label{eq:capacity_main}
    &C_{fb}(P)= \max_{\Pi,\hat{\Sigma},\Gamma} \frac1{2}\log \det (\Psi_Y) - \frac1{2}\log\det(\Psi)\nn\\
&\ \ \ \text{s.t. } \ \ \ \Psi_Y=\Lambda \Pi \Lambda^T + H \hat{ \Sigma} H^T+ \Lambda \Gamma H^T+ H \Gamma^T\Lambda^T  + \Psi \nn\\
&\begin{pmatrix}
    \Pi & \Gamma\\
    \Gamma^T& \hat{ \Sigma}
\end{pmatrix} \succeq0, \ \ \mathbf{Tr}(\Pi)\le P, \nn \\
&\begin{pmatrix}
     F \hat{ \Sigma} F^T + K_{p}\Psi K_{p}^T - \hat{ \Sigma}& F (\Gamma^T \Lambda^T + \hat{ \Sigma} H^T) + K_p\Psi \\
    (\Lambda \Gamma + H \hat{ \Sigma}) F^T + \Psi K_p^T & \Psi_Y
    \end{pmatrix} \succeq0,
\end{align}
where $K_p$ and $\Psi$ are constants given in \eqref{eq:ricc_constants}, \textcolor{black}{and the optimization variables are the matrices $\Pi\in\mathbb{R}^{m\times m}$,  $\hat{\Sigma}\in\mathbb{R}^{n\times n}$, and~$\Gamma\in\mathbb{R}^{m\times n}$}.
\end{theorem}
Note that by the first LMI constraint, the optimization variables $\Pi$ and $\hat{\Sigma}$ are positive semidefinite. The objective structure is the difference between the entropy rates of the channel outputs process and that of the noise process. Note that the objective is a concave function of the decision variables since $\Psi_Y$ is a linear function the decision variables, while $\Psi$ is a constant. Thus, the optimization problem is a convex optimization that can be computed with standard software e.g. \cite{cvx}.

In Section \ref{subsec:inputs}, the decision variables will be given a straightforward interpretation by showing that they induce a time-invariant, optimal inputs distribution. Here, we briefly remark on the LMIs in \eqref{eq:capacity_main}. The decision variable $\Pi$ corresponds to the covariance matrix of the channel input so that the first LMI in \eqref{eq:capacity_main} is a verification that it forms a valid covariance matrix with a correlated variable whose covariance matrix is $\hat{\Sigma}$. For the second LMI in \eqref{eq:capacity_main}, its Schur complement implies the Riccati inequality
\begin{align}\label{eq:Schur_Riccati_main}
     \hat{ \Sigma}&\preceq F \hat{ \Sigma} F^T + K_{p}\Psi K_{p}^T - K_Y\Psi_YK_Y^T,
\end{align}
with $K_Y = (F (\Gamma^T \Lambda^T + \hat{ \Sigma} H^T) + K_p\Psi)\Psi_Y^{-1}$. In Lemma \ref{lemma:2conditions} (Section \ref{sec:derivation}), it is shown that there always exist optimal decision variables $(\Pi,\hat{\Sigma},\Gamma)$ that satisfy the Riccati inequality \eqref{eq:Schur_Riccati_main} with equality, i.e., it is a Riccati equation. This reveals that the origin for explicit capacity formulae expressed as function of roots to some polynomials \cite{Kim06_MA,Kim10_Feedback_capacity_stationary_Gaussian,YangKavcicTatikondaGaussian} is the Riccati equation. We demonstrate this interesting fact in Section \ref{subsec:MA} for the MA noise process.

If the channel outputs, inputs, and the additive noise are scalars, but the hidden state of the noise $\vs_i$ can still be a vector, the capacity in Theorem \ref{th:main} can be simplified as follows.
\begin{theorem}[Feedback capacity of scalar channels]\label{th:scalar}
The feedback capacity of the scalar Gaussian channel \eqref{eq:pre_channel}-\eqref{eq:noise_state_space} with $\Lambda=1$ is given by the convex optimization problem
\begin{align}\label{eq:theorem_scalar}
   & C_{fb}(P)= \max_{\hat{\Sigma},\Gamma} \frac1{2}\log  \left( 1+ \frac{P + H \hat{ \Sigma} H^T+ \Gamma H^T  + H\Gamma^T}{\Psi}\right) \nn\\
&\text{s.t. } \ \ \ \begin{pmatrix}
    P & \Gamma\\
    \Gamma^T& \hat{ \Sigma}
\end{pmatrix} \succeq0, \ \ \\
&\begin{pmatrix}
     F \hat{ \Sigma} F^T + K_{p}\Psi K_{p}^T - \hat{ \Sigma}& F \Gamma^T + F \hat{ \Sigma} H^T + K_p\Psi \\
    \Gamma F^T + H\hat{ \Sigma} F^T + \Psi K_p^T & P  +  H\hat{ \Sigma} H^T + \Gamma H^T  + H\Gamma^T + \Psi
    \end{pmatrix} \succeq0,\nn
\end{align}
where $K_p$ and $\Psi$ are constants given in \eqref{eq:ricc_constants}.
\end{theorem}
Choosing $H=0$ in \eqref{eq:theorem_scalar} recovers the capacity formula of an additive white Gaussian noise (AWGN) channel $$C_{fb}(P) = \frac1{2}\log\left(1+\frac{P}{V}\right).$$

\begin{remark}\label{remarkgattami}
The state-space studied in \cite{Gattami} can be recovered from the setting in Theorem \ref{th:scalar} with $W=V=L=1$ and a stable $F$. In this case, the constants are $\Sigma=0,K_p= G, \Psi =1$ and the capacity expression in \eqref{eq:theorem_scalar} and that in \cite[Th. $4$]{Gattami} are almost in full agreement. In particular, they write the optimization problem with a supremum, and there is a difference in the sign of the first LMI in \eqref{eq:theorem_scalar} which reads as a strict LMI ($\succ$) in \cite{Gattami}. The reasoning for their strict LMI follows from an erroneous claim after their Theorem $3$ that the i.i.d. component of the optimal policy should have a positive variance at all times (the policy structure appears below in \eqref{discuss:policy}). This claim is utilized to show their argument on the invertibility of $\hat{\Sigma}$. In Section \ref{subsec:MA} we show, for the MA noise, that \textcolor{black}{the i.i.d. component of the optimal policy is zero}. In particular, we show that the optimum is achieved on the boundary of the LMI, i.e., the optimal variables induce a singular matrix in the LMI (see the proof of Theorem \ref{th:MA} for details).
\end{remark}



\subsection{Moving average (MA) noise}\label{subsec:MA}
In this section, we study a scalar channel with the MA noise process
\begin{align}\label{eq:MA_discussion}
    z_i&= w_i + \alpha w_{i-1},
\end{align}
with i.i.d. $w_i\sim N(0,1)$ and $\alpha\in\mathbb{R}$. In \cite{Kim06_MA}, the feedback capacity of the MA noise with $|\alpha|\le1$ was shown to be equal to
\begin{align}\label{eq:KIM_MA_capa}
 C_{fb}(P) &= -\log x_0,
\end{align}
where $x_0$ is the unique positive root of
\begin{align}\label{eq:KIM_poly}
Px^2 = (1- |\alpha| x)^2(1-x^2).
\end{align}
The MA noise can be realized by the state space \eqref{eq:noise_state_space} with $F=0,H=\alpha,G=W=V=L=1$. We derive here the feedback capacity expression for all $\alpha$.
\begin{theorem}[Moving-average noise]\label{th:MA}
The feedback capacity of the Gaussian channel with MA noise is
\begin{align}\label{eq:th_capacity_MA}
    C_{fb}(\alpha,P)&= \frac1{2}\log(1+\mathbf{SNR}),
\end{align}
where $\mathbf{SNR}$ is the maximal positive root of the polynomial
\begin{align}\label{eq:ma_our_poly}
    \mathbf{SNR} = \begin{cases}  \left(\sqrt{P} + |\alpha|\sqrt{\frac{\mathbf{SNR}}{1+\mathbf{SNR}}} \right)^2 & \text{if} \ \ |\alpha|\le1\\
      |\alpha|^{-2}\left( \sqrt{P} + \sqrt{\frac{\mathbf{SNR}}{1+\mathbf{SNR}}} \right)^2 & \text{if} \ \ |\alpha|> 1 \ \text{and} \ \Sigma_1>0.
    \end{cases}
\end{align}
\textcolor{black}{The capacity expression in \eqref{eq:th_capacity_MA}-\eqref{eq:ma_our_poly} for $|\alpha|\le1$ coincides with the feedback capacity expression in~\eqref{eq:KIM_MA_capa}.}
\end{theorem}
\textcolor{black}{For $|\alpha|\le1$, the feedback capacity is independent of the initial state covariance $\Sigma_1$ but, for $|\alpha|>1$, we assume $\Sigma_1\neq0$. This condition is made to avoid the singularity that $\Sigma_i=0$ for all $i$, and does not converge to the stabilizing solution of the Riccati equation. If $|\alpha|>1$ and $\Sigma_1=0$ (i.e., the initial state is deterministic), the capacity can still be computed but with the solution to the first polynomial in \eqref{eq:ma_our_poly}.}

To compare the capacity expression in Theorem \ref{th:MA} with \eqref{eq:KIM_MA_capa} when $|\alpha|\le1$, we can use the change of variable $x^2 =  (1+\mathbf{SNR})^{-1}$ to write \eqref{eq:th_capacity_MA} as $C_{fb}(P) = -\log x_0$ where $x_0$ solves $\frac{1-x^2}{x^2}=\left(\sqrt{P} + |\alpha|\sqrt{1-x^2} \right)^2$. It is interesting to note that the latter polynomial and \eqref{eq:KIM_poly} are different. However, the second part of Theorem \ref{th:MA} confirms that the feedback capacities are in agreement for $|\alpha|\le1$ by showing that their positive roots coincide. We proceed to prove Theorem~\ref{th:MA}.
\begin{proof}[Proof of Theorem \ref{th:MA}]
We compute the capacity expressions in Theorem \ref{th:scalar}, and then verify thereafter that the required conditions are met.

The Schur complement of $\begin{pmatrix}
    P & \Gamma\\
    \Gamma^T& \hat{ \Sigma}
\end{pmatrix}\succeq0$ evaluated in the optimal variables is shown to be zero using contradiction. Assume that $P-\Gamma^2\hat{\Sigma}^{-1}=p$ for some $p>0$. Then, we can replace $\Gamma$ with $\Gamma'=\Gamma(1+\Gamma^{-2}p\hat{\Sigma})^{1/2}$ to obtain a larger objective. The Riccati LMI can be verified to be satisfied with this replacement. This step shows that the LMI cannot be strict at the optimum. For the other LMI in Theorem \ref{th:scalar}, a similar reasoning shows that the Schur complement of the Riccati LMI \eqref{eq:Schur_Riccati_main} can be achieved with equality (see Lemma \ref{lemma:2conditions}), and the Schur complement simplifies to the Riccati equation $\hat{\Sigma} =  K_p \Psi K_p - (\Psi K_p)^2\Psi_Y^{-1}$.

\textcolor{black}{In order to compute the capacity expression in Theorem \ref{th:scalar}, we compute the Riccati constants $K_p$ and $\Psi$ from the stabilizing solution of the Riccati equation in \eqref{eq:riccati_general}. The Riccati equation has two solutions $\Sigma=0,1-\frac1{\alpha^2}$. For $|\alpha|<1$, the stabilizing solution is $\Sigma=0$ which implies $K_p=\Psi=1$, while for $|\alpha|>1$, the stabilizing solution is $\Sigma  = 1 - \frac1{\alpha^2}$ which implies $\Psi = \alpha^2$ and $K_p = \alpha^{-2}$.} We note that in both cases $K_p\Psi=1$ so that the Riccati equation above simplifies to $\hat{\Sigma} = K_p - \Psi_Y^{-1}$ where $K_p$ is either $1$ or $\alpha^{-2}$. The decoder innovation can be written as
\begin{align}\label{eq:ma_proof_fixed}
    \Psi_Y = \Psi + P + \alpha^2\hat{\Sigma}  + 2|\alpha|\sqrt{P\hat{\Sigma}} = \Psi + \left(\sqrt{P} + |\alpha|\sqrt{K_p - \Psi_Y^{-1}}\right)^2,
\end{align}
where the sign of $\Gamma$ was chosen to maximize $\Psi_Y$. We denote $\Psi_Y\Psi^{-1} = 1+\mathbf{SNR}$ and substitute the latter in both sides of \eqref{eq:ma_proof_fixed} to obtain the fixed-point equations in \eqref{eq:ma_our_poly}.

We verify the conditions for Theorem \ref{th:scalar}. The pair $(F=0,H=\alpha)$ is detectable (Assumption $1$) for all $\alpha$, and $(F_s,W_s)=(-\alpha,0)$ is controllable on the unit-circle for all $|\alpha|\neq 1$ (Assumption $2$). Thus, for $|\alpha|\neq1$, the stabilizing solution for the Riccati equation exists and is equal to $\Sigma=\max\{0,1-\frac1{\alpha^2}\}$. It is easy to check that the Riccati recursion in \eqref{eq:Riccati_recursion}, $\Sigma_{i+1} = 1 - \frac1{1+\alpha^2\Sigma_i}$, converges to the stabiliing solution unless $|\alpha|>1$ and $\Sigma_1=0$.

\textcolor{black}{If $|\alpha|=1$, the only solution to the Riccati equation is $\Sigma=0$, but it is not a stabilizing solution (it is the maximal solution). Although Theorem \ref{th:main} concerns with noise processes whose Riccati equations have stabilizing solutions, the upper bound extends to non-stabilizing solutions as well, and for the particular instance of the MA noise with $|\alpha|=1$, we verified that the lower bound in Lemma \ref{lemma:achievable} holds as well.}

Finally, we show the equivalence of our capacity expression and \eqref{eq:KIM_MA_capa}, for $|\alpha|<1$, by proving that the positive roots of \eqref{eq:KIM_poly} and $\frac{1-x^2}{x^2}=\left(\sqrt{P} + |\alpha|\sqrt{1-x^2} \right)^2$ coincide. The positive root of \eqref{eq:KIM_poly} satisfies $\frac{1-x_0^2}{x_0^2} = \frac{P}{(1- |\alpha| x_0)^2}$ and $\sqrt{1-x_0^2} = \frac{\sqrt{P}x_0}{1-|\alpha|x_0}$, and by substituting these equations into the second polynomial, we get $$\frac{1-x^2}{x^2} - \left(\sqrt{P} + |\alpha|\sqrt{1-x^2} \right)^2 = \frac{P}{(1- |\alpha| x_0)^2} - \left(\sqrt{P} + |\alpha|\frac{\sqrt{P}x_0}{1-|\alpha|x_0} \right)^2=0.$$
The other direction can be shown similarly.
\end{proof}

\subsection{Auto-regressive (AR) noise}
The auto-regressive (AR) noise process of first order is given by \begin{align}\label{eq:AR_discussion}
    z_i&= \beta z_{i-1} + w_i,
\end{align}
where $w_i\sim N(0,1)$ is an i.i.d. sequence. This is one of the simplest instances of colored Gaussian noise and was studied in \cite{Butman_conjecture,Butman69,TiernanSchalk_AR_UB}. \textcolor{black}{A closed-form feedback capacity expression for the stationary case $|\beta| <1$ was derived in \cite{Kim10_Feedback_capacity_stationary_Gaussian}}\footnote{\textcolor{black}{Recently \cite{comments_kim} showed that \cite{Kim10_Feedback_capacity_stationary_Gaussian} has an error in Corollary $4.4$, which renders the feedback capacity expression in \cite{Kim10_Feedback_capacity_stationary_Gaussian} for the AR noise with $|\beta|<1$ unjustified. Particularizing our general capacity expression to AR noise and carrying out a computation similar that of the MA noise (Theorem 3) shows that the feedback capacity expression in \cite{Kim10_Feedback_capacity_stationary_Gaussian} is correct for $|\beta| < 1$.}}. We present next the AR noise with a general $\beta$. The AR process can be realized by \eqref{eq:noise_state_space} with $F = H= \beta$ and $G=L=W=V=1$.

\begin{figure}[t]
    \centering
\psfrag{B}[t][][1]{Regression parameter $\beta$}
\psfrag{A}[b][][1]{Rate [bits/ch. use]}
\includegraphics[scale=0.3]{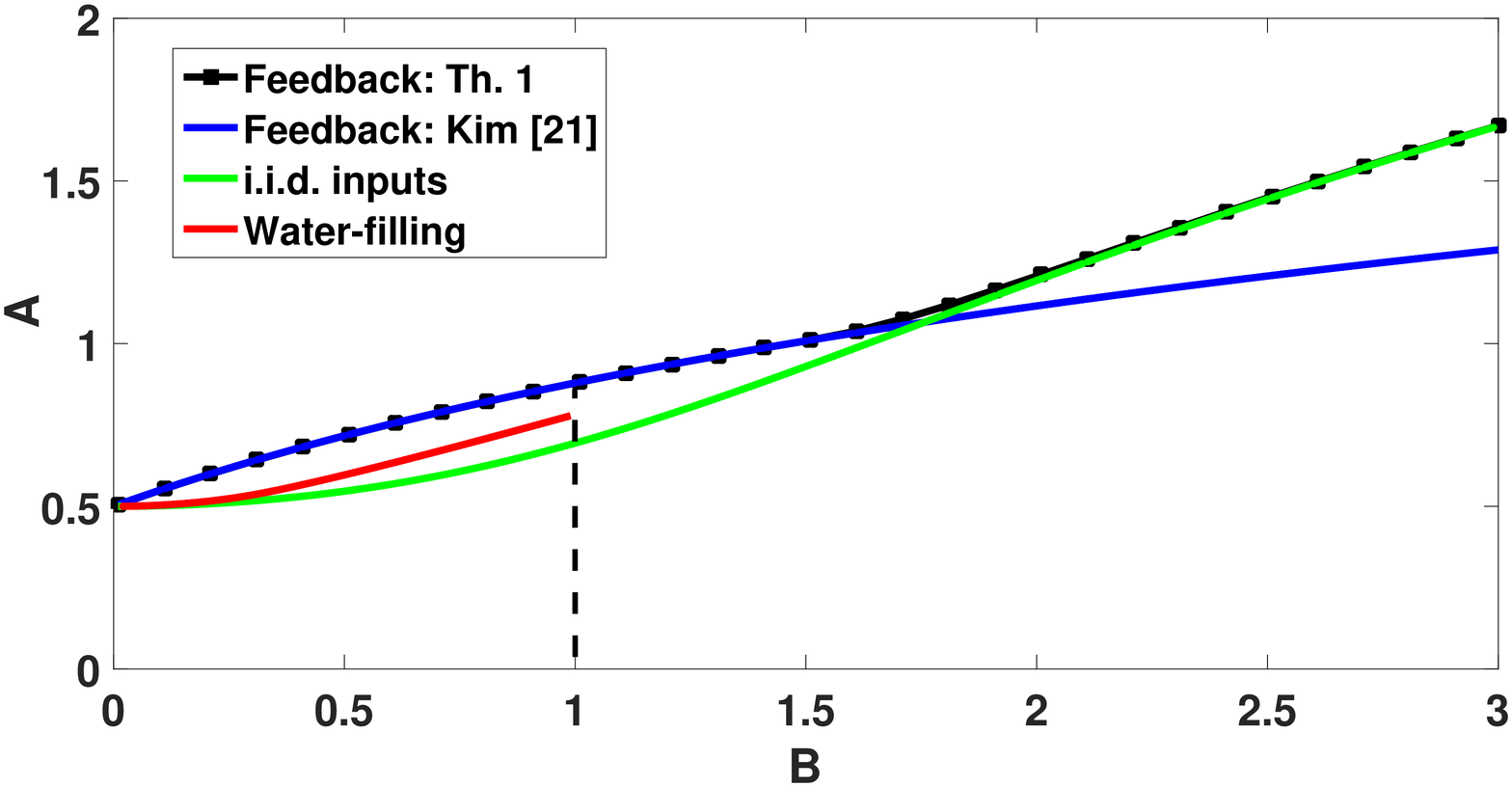}
    \caption{The feedback capacity of the Gaussian channel with an auto-regressive noise of first order and a unit-power input constraint (black curve). The blue curve describes the feedback capacity expression for the stationary case ($|\beta|<1$) from \cite{Kim10_Feedback_capacity_stationary_Gaussian}, but is plotted here for greater values of $\beta$, \rr{as it numerically coincides with the feedback capacity as long as $M=0$ (see Fig. \ref{fig:AR_power} below)}. The red curve corresponds to the feedforward capacity (without feedback) obtained via a water-filling solution, and the green curves corresponds to an i.i.d. coding law (feedback-independent) of the channel inputs \rr{in \eqref{eq:Riid}}.}
    \label{fig:AR_capa}
\end{figure}

In Fig.~\ref{fig:AR_capa}, the feedback capacity in Theorem \ref{th:main}, the feedback capacity expression for $|\beta|<1$ from \cite{Kim10_Feedback_capacity_stationary_Gaussian}, and the non-feedback capacity (using a water filling solution) are plotted. Additionally, we plot the maximal achievable rate with i.i.d. inputs by adding the constraint $\Gamma=0$ to \eqref{eq:theorem_scalar}. This rate can be computed explicitly as
\begin{align}\label{eq:Riid}
  R_{iid}(\beta,P=1) =\frac1{2}\log\left(1+ \frac{\beta^2}{2}\left(1 + \sqrt{1+\frac{4}{\beta^4}}\right)\right).
\end{align}
First, it is interesting to note that black and blue curves coincide for some non-stationary noise with $1\le \beta\lesssim 1.5$. This means that the capacity expression in \cite{Kim10_Feedback_capacity_stationary_Gaussian} holds true even for some values beyond the stationary regime. The rate achieved with i.i.d. inputs (green curve) approaches the feedback capacity for large regression parameters. Thus, the plot shows that, as the regression parameter grows large, the rate achieved by i.i.d. inputs approaches the feedback capacity, and the feedback link has a negligible contribution in terms of capacity. From operational perspective, the feedforward capacity lies in between the i.i.d. achievable rate (green curve) and the feedback capacity (black curve). Thus, for large $\beta$, it can be well-approximated with simple i.i.d. inputs, i.e., codewords with memory are not needed. These phenomena are related to the structure of the optimal inputs distribution that is presented and discussed for the AR noise in Section \ref{subsec:inputs}.

\section{Optimal inputs distribution and coding scheme}\label{sec:scheme}
In this section, we present a capacity-achieving, time-invariant inputs distribution that can be computed from the convex optimization in Theorem~\ref{th:main}. We then use this inputs distribution to construct a capacity-achieving coding scheme for scalar channels and discuss a possible extension to MIMO channels.

\subsection{Optimal inputs distribution}\label{subsec:inputs}
The optimal decision variables in Theorem \ref{th:main} induce a time-invariant capacity-achieving inputs distribution\footnote{More precisely, a capacity-achieving policy in Lemma \ref{lemma:achievable} is the time-invariant law in \eqref{discuss:policy} for $t>1$, and a different coding law for $t=1$.}:
\begin{align}\label{discuss:policy}
\vx_i&= \Gamma \hat{\Sigma}^\dagger (\vhs_{i}-\vhhs_{i}) + \vm_i,
\end{align}
where $\vhs_i$ is defined in \eqref{eq:KF_estimator}, and its estimate at the decoder is defined by
\begin{align}\label{eq:vhhs_def}
    \vhhs_{i}&\triangleq \E[\vhs_{i}|\vy^{i-1}].
\end{align}
The optimal policy is composed as the sum of two signaling components. The first component, $(\vhs_{i}-\vhhs_{i})$ corresponds to the decoder's estimation error, and is a function of the feedback. Its transmission refines the decoders' knowledge on $\vhs_i$ by transmitting the states innovation (the vector $\vhs_i$ can also be regarded as the channel state, see Section \ref{sec:setting}). \textcolor{black}{The covariance of the innovation $(\vhs_{i}-\vhhs_{i})$ is $\hat{\Sigma}$, so that the covariance of the first component is $\cov(\Gamma \hat{\Sigma}^\dagger (\vhs_{i}-\vhhs_{i}))=\Gamma \hat{\Sigma}^\dagger\Gamma^T$.} The second component, $\vm_i$, is independent of $(\vx^{i-1},\vy^{i-1})$ (and thus is feedback-independent), and has an i.i.d. distribution with the remaining covariance, i.e., $\vm_i\sim N(0,M)$ with $M\triangleq\Pi - \Gamma \hat{\Sigma}^\dagger\Gamma^T$. The transmission of the vector $\vm_i$ increases the uncertainty of the channel state $\vhs_i$ at the decoder\textcolor{black}{, but it can be used to transmit new information (on the message).} In the AWGN channel, for instance, the entire power is allocated to the second component $\vm_i$. We proceed to illustrate the policy behavior for the AR noise.

\begin{figure}[b]
    \centering
\psfrag{B}[l][][1]{$\Gamma \hat{\Sigma}^\dagger \Gamma^T$}
\psfrag{A}[l][][1]{$M$}
\psfrag{C}[][][1]{Regression parameter $\beta$}
\psfrag{D}[][][1]{Power}
\includegraphics[scale=0.2]{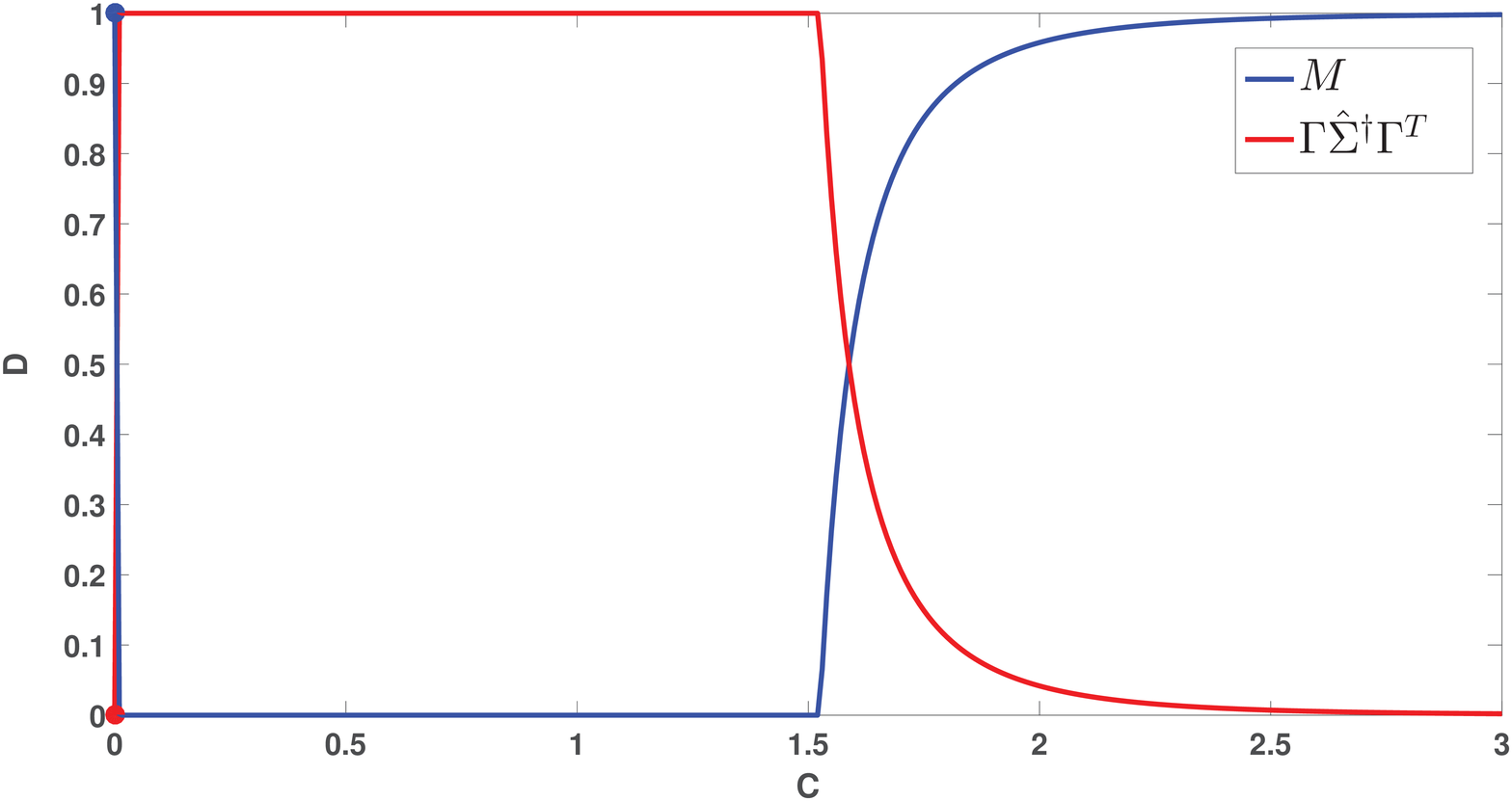}
    \caption{The power of the signaling components in the inputs distribution \eqref{discuss:policy} for AR noise. The red curve corresponds to the feedback-dependent signal $\Gamma\hat{\Sigma}^\dagger(\vhs_i-\vhhs_i)$ and the blue curve corresponds to $\vm_i$. Note that there are two phase transitions, at $\beta=0$ and $\beta\approx1.5$. Consistent with \cite[Th. $4.6$]{Kim10_Feedback_capacity_stationary_Gaussian}, $M=0$ achieves the feedback capacity in the stationary regime $\beta\in(0,1)$. \textcolor{black}{Also, as $\beta$ grows large, the power allocated to the i.i.d. component grows as well. Combined with the i.i.d. coding law curve in Fig. \ref{fig:AR_capa} (green curve), this shows that the feedback link has negligible contribution to the capacity solution.}}
    \label{fig:AR_power}
\end{figure}

In Fig.~\ref{fig:AR_power}, the power allocated to each of the signals in \eqref{discuss:policy} is plotted for the AR noise in \eqref{eq:AR_discussion} with a power constraint $P=1$. \textcolor{black}{For $0<\beta\le1$, Fig.~\ref{fig:AR_power} agrees with the claim in \cite[Th. $4.6$]{Kim10_Feedback_capacity_stationary_Gaussian} on the sufficiency of inputs distribution with $\vm_i=0$  for sclar and stationary noise (see also next paragraph).} However, beyond the stationary regime, there is a sharp phase transition, and the power allocated to $\vm_i$ increases as $\beta$ grows. The phase transition location, $\beta\approx 1.5$, explains the gap between the feedback capacity in Theorem \ref{th:main} and the capacity expression in \cite{Kim10_Feedback_capacity_stationary_Gaussian} since the latter used a policy with $\vm_i=0$. Fig. \ref{fig:AR_power} also shows that the rate achieved with i.i.d. inputs in \eqref{eq:Riid} approaches the feedback capacity for growing $\beta$. This implies that the Schalkwijk-Kailath (SK) encoding law is close to optimal in this regime \cite{SchalkwijkKailath66_feedback_scheme}.

\textcolor{black}{The role of the second component $\vm_i$ has been discussed in several papers \cite{Kim10_Feedback_capacity_stationary_Gaussian,Gattami,comments_kim}. In \cite[Cor. $4.4$]{Kim10_Feedback_capacity_stationary_Gaussian}, it is claimed that for scalar channels with stationary noise, the capacity can be achieved with $M=0$. Recently, \cite{comments_kim} showed that the proof of the claim in \cite[Cor. $4.4$]{Kim10_Feedback_capacity_stationary_Gaussian} relies on an erroneous calculation and thus is invalid. Our capacity derivation relies on a general policy with $M\succeq0$ (with a different coding for the first time $t=1$). As illustrated in the examples above, for general noise processes, $M$ can be either positive or zero; for the MA noise, we prove in Theorem \ref{th:MA} that $M=0$ is necessary to achieve the capacity, and for the AR noise, it is illustrated that $M>0$ in the non-stationary regime (Fig. \ref{fig:AR_power}). When specializing our capacity expression for stationary noise processes, it may be utilized to find a counterexample for \cite[Cor. $4.4$]{Kim10_Feedback_capacity_stationary_Gaussian}.
We ran extensive simulations to specialize our capacity expression in Theorem \ref{th:scalar} to various stationary noise processes and to compute the optimal $M$, yet we did not find a counterexample to [21, Cor. 4.4]. Thus the claim in \cite[Cor. $4.4$]{Kim10_Feedback_capacity_stationary_Gaussian} may be true.}

As mentioned in Remark \ref{remarkgattami}, the fact that $M=0$ does not imply that the achievable rate is as erroneously claimed in \cite{Gattami}. If $M=0$, it simply implies that the message is encoded at the first time with $\vm_1\neq0$, and from $t>1$ the encoder follows the rule $\vx_i=\Gamma\hat{\Sigma}^\dagger(\vhs_i-\vhhs_i)$. In the next section, we show that this explicit coding scheme is capacity-achieving with double-exponential decay in the error probability for any rate below capacity.

\subsection{Coding scheme for scalar channels}
In this section, we construct a capacity-achieving coding scheme for scalar channels (with a vector hidden state) based on the optimal inputs distribution in \eqref{discuss:policy}. Throughout this section, it is assumed that the optimal inputs distribution in \eqref{discuss:policy} satisfies $M = 0$. The design of an explicit coding scheme with $M\neq0$ remains open. 

Our scheme resembles the SK scheme \cite{SchalkwijkKailath66_feedback_scheme,gallager_SK_simple}, and other posterior matching schemes for memoryless channels \cite{horstein_original,shayevitz_posterior_mathcing,tara_posterior,Li_elgamal_matching} and channels with memory \cite{Kim06_MA,Kim10_Feedback_capacity_stationary_Gaussian,Sabag_BIBO_IT,Sabag_BEC,achilles_posterior_scheme,PeledSabagBEC,Elia_MIMO_ITW} in its main idea to refine the decoders' knowledge of the message (or equivalently, to refine the decoders' knowledge of the first noise instance $\vz_0$ in Gaussian channels). The main difference with the SK scheme is that rather than encoding the scaled innovation of the first noise instance $z_0$, our encoding follows \eqref{discuss:policy} to transmit the innovation of $\vhs_i$. This modification results in a more numerically stable encoding since our scaling factor $\Gamma\hat{\Sigma}^\dagger$ is a constant, while in the SK scheme the scaling of the message innovation increases with time.

A related scheme for a similar setting appears in \cite{Kim10_Feedback_capacity_stationary_Gaussian}. Both schemes follow the encoding in \eqref{discuss:policy} but, only our paper provides a computable expression for the coefficient matrix $\Gamma\hat{\Sigma}^\dagger$ (via Theorem \ref{th:main}). Additionally, we simplify the multidimensional encoding method in \cite{Kim10_Feedback_capacity_stationary_Gaussian} by showing that, even when the hidden state is a vector, it is sufficient to encode the message in a single time instance. Additionally, we provide an explicit smoother for the maximum likelihood decoder proposed in \cite{Kim10_Feedback_capacity_stationary_Gaussian,Ihara_Gaussian_Coding}.



Our coding scheme consists of three building blocks:
\begin{enumerate}
\item The message $m\in[1: 2^{nR}]$ is mapped to a zero-mean, unit-separation symbol $U(m) = m-2^{nR-1}$ whose variance is $\text{Var}(U) = (2^{2nR}-1)/12$. The normalized symbol $\bar{U}(m) = \text{Var}(U)^{-\frac1{2}}U(m)$ will be transmitted at the first time instance as $\vx_0(m)$.
\item In subsequent times, the encoding process is simply to use the optimal inputs distribution in \eqref{discuss:policy}. The estimates $\vhs_i$ and $\vhhs_i$ can be computed directly from \eqref{eq:prel_estimator} and \eqref{eq:vhhs_recursion} and the constant $\Gamma\hat{\Sigma}^\dagger$ is obtained from the optimization in Theorem \ref{th:main}.
\item  When transmission ends, the decoder constructs the maximum-likelihood estimate of the first noise instance $\vz_0$ using the measurements $\vy^n$. This is a smoothing problem that is formally presented in Lemma \ref{lemma:smoother}.
\end{enumerate}
We are now ready to present the coding scheme as Algorithm~\ref{alg:bounded}. The abbreviation $\text{KF}$ stands for Kalman filter while $\text{Smooth}$ is for the smoothing function in Lemma \ref{lemma:smoother} (Eq. \eqref{eq:smoother}).

\begin{algorithm}[h!]
\caption{Optimal scheme for scalar channels}
\label{alg:bounded}
\begin{algorithmic}
\State Inputs: $m$, $\Gamma$, $\hat{\Sigma}$, $\vhs_0=\vhhs_0=0$, \rr{$\hat{z}_{0|0}=0$}
\State $x_0\gets \bar{U}(m)$ \Comment \emph{Transmission}
\State Store $y_0 = x_0+z_0$ \Comment \emph{Dec.}
\Procedure{$i=1:n$}{}
\State $\vhs_{\rr{i}}\gets \text{KF}(\vhs_{\rr{i-1}},z_{i-1} = y_{\rr{i-1}}-x_{\rr{i-1}})$
 \Comment \emph{Enc. (Eq. \eqref{eq:enc_KF_LTI})}
\State $\vhhs_{\rr{i}} \gets \text{KF}(\vhhs_{\rr{i-1}},y_{\rr{i-1}},i\rr{-1})$
 \Comment \emph{Dec. (Eq. \eqref{eq:vhhs_recursion})}
\State $x_{\rr{i}} \gets \Gamma \hat{\Sigma}^\dagger(\vhs_{\rr{i}} - \vhhs_{\rr{i}})$  \Comment \emph{Transmission}
\State $y_{\rr{i}} = x_{\rr{i}} + z_i$  \Comment \emph{Ch. output}
\State $\hat{z}_{\rr{0|i}}\gets \text{Smooth}(\hat{z}_{\rr{0|i-1}},y_{\rr{i}},i)$ \Comment \emph{Dec. (Eq. \eqref{eq:smoother})}
\EndProcedure
\State $\hat{m} = \arg\min_{m'} |(y_0 - \hat{z}_{\rr{0|n}}) - \bar{U}(m')|$ \Comment \emph{Dec.}
\end{algorithmic}
\end{algorithm}
\textcolor{black}{We remark that the encoder estimate $\vhs_i$ is computed with the time-invariant Kalman filter in  \eqref{eq:enc_KF_LTI}, while the decoder Kalman filter for $\vhhs_i$ is time-varying is used with the time-invariant parameters $\Psi_i=\Psi, K_{p,i}=K_p$ and the initial conditions $\rr{\hat{\Sigma}_0}=0$ and $M_0 = V(\bar{U})$.}
The following theorem shows the optimality of our coding scheme.
\begin{theorem}[Capacity-achieving coding scheme]\label{th:coding}
For any rate $R<C_{fb}(P)$, the error probability of the coding scheme for scalar channels (with $\vm_i=0$) in Algorithm~\ref{alg:bounded} decays in a doubly-exponential rate for large $n$.
\end{theorem}
A simple proof of Theorem \ref{th:coding} appears at the end of this section. Its main building block is the analysis of a smoothing problem in Lemma \ref{lemma:smoother}. We provide now explicit formulas to compute the estimate and its error covariance. The formulas are presented for the general MIMO channel, and their particularization to scalar channels will be used in the proof of Theorem~\ref{th:coding}.
\begin{lemma}[The smoothing problem]\label{lemma:smoother}
Consider the smoothing problem of estimating $\vz_0$ from $\vy^n$ with
\begin{align}
 \hat{\vz}_{0|n}&\triangleq \E[\vz_0|\vy^n] \nn\\
 \hat{Z}_{0|n}&\triangleq \cov(\vz_0 - \hat{\vz}_{0|n}).
\end{align}
Subject to the optimal inputs distribution \eqref{discuss:policy}, when $\vm_i=0$,
\begin{enumerate}
    \item The optimal smoother can be recursively computed as
    \begin{align}\label{eq:smoother}
        \hat{\vz}_{0|i}&= \hat{\vz}_{0|i-1} + \hat{Z}_{0|i-1}\kappa_i^T  \Psi_{Y,i}^{-1}(\vy_i-H\vhhs_{i}),
    \end{align}
with $\hat{\vz}_{0|0}=0$ and $$\kappa_i =  (\Lambda\Gamma\hat{\Sigma}^\dagger + H )(F - K_p(\Lambda\Gamma\hat{\Sigma}^\dagger + H))^{i-1} K_p.$$
    \item The error covariance can be updated as
    \begin{align}\label{eq:lemma_cov_reduction}
               \hat{Z}_{0|i} &= (I - \hat{Z}_{0|i-1}\kappa_i^T \Psi_{Y,i}^{-1} \kappa_i)\hat{Z}_{0|i-1}
    \end{align}
    with $\hat{Z}_{0|0} = \Psi$, and its determinant satisfies
    \begin{align}\label{eq:lemma_cov_reduction_det}
        \det (\hat{Z}_{0|i})&= \det (\Psi_{Y,i}^{-1}\Psi) \det (\hat{Z}_{0|i-1}).
    \end{align}
    Moreover, $\Psi_{Y,i}$ converges to $\Psi^\ast_Y$, the optimal value of $\Psi_Y$ in Theorem \ref{th:main}.
    \item[-] Therefore, for scalar channels, the error covariance satisfies
\begin{align}\label{eq:lemma_smooth_scalar}
        \hat{Z}_{0|i} &= \Psi_{Y,i}^{-1}\Psi\hat{Z}_{0|i-1},
\end{align}
and $\Psi_{Y,i}$ converges to $\Psi^\ast_Y$.
\end{enumerate}
\end{lemma}
The proof of Lemma \ref{lemma:smoother} appears in Section \ref{subsec:proofs_coding}. Recall from Theorem \ref{th:main} that the capacity can be expressed as $C_{fb}(P) = \frac1{2}\log\det (\Psi_Y^\ast\Psi^{-1})$ where $\Psi_Y^\ast$ denotes the optimal $\Psi_Y$. \textcolor{black}{The relation between the capacity and the smoothing problem is transparent.} The volume (determinant) reduction of the error covariance in \eqref{eq:lemma_cov_reduction_det} is the logarithm argument in the capacity expression. For scalar channels, the volume reduces to a single dimension refinement of the noise instance $\vz_0$. However, for MIMO channels the volume reduction is not sufficient to derive an explicit coding scheme. We provide details on a possible MIMO scheme construction.

A suggested scheme for MIMO channels is as follows: assume for simplicity $\Lambda=I$, and split the message into $p$ \textcolor{black}{independent sub-messages} where $p$ is dimension of the input, output and noise. \textcolor{black}{In the first time, we normalize each sub-message, and transmit their concatenation as the vector $\vx_0$}. The encoding is identical to that in Algorithm~\ref{alg:bounded}, that is, it follows the policy in \eqref{discuss:policy}. \textcolor{black}{The estimation of the vector $\vz_0$ is based on the smoother in \eqref{eq:smoother},} and the decoding is carried out using coordinate-wise successive cancellation of the vector $\hat{\vz}_{0|n}$. The analysis of such scheme can be possibly done using \textcolor{black}{Lemma \ref{lemma:smoother}}, but a finer spectral analysis is needed. \textcolor{black}{In particular, the geometric reduction in \eqref{eq:lemma_cov_reduction_det} needs to be shown for each coordinate and not for the overall determinant}. The geometric rate of the error covariance in \eqref{eq:lemma_smooth_scalar} should also determine the rates allocated to each sub-message, and is the key to obtain the double-exponential decay. We proceed to show the optimality of Algorithm~\ref{alg:bounded} for scalar channels using the analysis of Lemma \ref{lemma:smoother}.
\begin{proof}[Proof of Theorem \ref{th:coding}]\label{proof:coding}
The decoder estimates $x_0$ from $y_0 - \hat{z}_{0|n} = \bar{U}(m) + z_0 - \hat{z}_{0|n}$. This is the problem of estimating an M-PAM signal from a Gaussian-corrupted measurement, and the probability of error can be bounded as
\begin{align}
        P_e^{(n)}&\le \Pr\left( |z_0 - \E[z_0|y^n]| \ge \frac1{2}\sqrt{\frac{12}{2^{2nR}-1}}\right) \nn \\
        &= 2 Q( \gamma_n),
\end{align}
where the inequality follows from the first and last messages where a large error deviation will not incur an error on one of their ends, and the equality follows from $\gamma_n \triangleq \frac1{2}\sqrt{12 \frac{\text{Var}(z_0|y^n)}{2^{2nR}-1}}$ with the standard $Q$-function. We can further bound $\gamma_n$ as
\begin{align}
    \gamma_n &\le \sqrt{3} \cdot 2^{-nR}\sqrt{\text{Var}(z_0|y^n)}.
\end{align}
By Lemma \ref{lemma:smoother}, $\gamma_n$ has a positive exponent if $R< \frac1{2n}\log\Psi \prod_{i=1}^n (\Psi_{Y,i}\Psi^{-1})$, which leads to the doubly-exponential decay rate. As $\Psi_{Y,i}\to\Psi_Y^\ast$, $R$ can be chosen arbitrarily close to $\frac1{2}\log(\Psi_{Y}^\ast\Psi^{-1})$ which is precisely the feedback capacity.
\end{proof}

\section{Proof of the main result}\label{sec:derivation}
In this section we outline the proof of Theorem \ref{th:main} by presenting the technical lemmas leading to tight lower and upper bounds. The proof is structured as three parts.

1. \textbf{Sequential convex optimization problem (SCOP)}: The $n$-letter capacity expression for the MIMO channel is defined as
\begin{align}\label{eq:objectiveN}
    C_n(P)&=\max_{P(\vx^n||\vy^n): \frac1{n}\sum_{i=1}^n\E[\vx_i^T\vx_i]\le P} h(\vy^n) - h(\vz^n).
\end{align}
The first three lemmas formulate the $n$-letter capacity as a SCOP. While it is easy to show that $C_n(P)$ is concave in its decision variable $P(x^n\|y^n)$, the challenge is to formulate it as a convex optimization problem that enables one to explicitly compute the limit of $C_n(P)$. To this end, we realize a SCOP with LMI constraints that have a sequential structure.

2. \textbf{Upper bound via convexity}: The second part of the proof utilizes the SCOP structure to show that the capacity expression in Theorem \ref{th:main} is an upper bound on the capacity. Since the optimization constraints contain decision variables at consecutive times, the standard time-sharing random variable argument does not apply here, and we use a different technique to show that these constraints are \emph{asymptotically satisfied} when realized at the convex combinations of the decision variables.

3. \textbf{Lower bound using time-invariant inputs}: The last part constructs a time-invariant policy whose optimization leads to a lower bound that is expressed as the upper bound optimization problem with additional constraints. We show that the additional constraints are redundant, concluding the proof of the main result.

\subsection{Sequential convex optimization problem}
The first lemma identifies an optimal structure for the inputs distribution using $\vhs_i$ and $\vhhs_{i}$ defined in \eqref{eq:KF_estimator} and \eqref{eq:vhhs_def}, respectively.
\begin{lemma}[The optimal policy structure]\label{lemma:policy}
For a fixed $n$, it is sufficient to optimize \eqref{eq:objectiveN} with inputs of the form
\begin{align}\label{eq:policy}
\vx_i&= \Gamma_i \hat{\Sigma}^\dagger_{i}(\vhs_{i}-\vhhs_{i}) + \vm_i,\ \ \ \  i=1,\dots,n
\end{align}
where $\vm_i\sim N(0,M_i)$ is independent of $(\vx^{i-1},\vy^{i-1})$, $\hat{\Sigma}^\dagger_{i}$ is the Moore-Penrose pseudo-inverse of
\begin{align}
    \hat{\Sigma}_{i}&= \cov( \vhs_{i}-\vhhs_{i}),
\end{align}
$\Gamma_i$ is a matrix that satisfies
\begin{align}\label{eq:lemma_policy_orthognality}
\Gamma_i(I - \hat{\Sigma}_{i}\hat{\Sigma}^\dagger_{i}) =0,
\end{align}
and the power constraint is
\begin{align}
    \frac1{n}\sum_{i=1}^n\mathbf{Tr}( \Gamma_i \hat{\Sigma}^\dagger_{i}\Gamma_i^T + M_i)\le P.
\end{align}
\end{lemma}
Lemma \ref{lemma:policy} simplifies the optimization \eqref{eq:objectiveN} by showing that the optimization domain is over the sequence of matrices $(\Gamma_i,M_i\succeq0)_{i=1}^n$. Note that $\hat{\Sigma}_{i}$ is a deterministic function of the policy up to time $i-1$ and thus is not part of the optimization. Similar policies appeared in the literature e.g. \cite[Section IV]{Kim06_MA} and \cite{Gattami} building on the ideas in \cite{YangKavcicTatikondaGaussian}. Their policy reads $\vx_i = \Gamma_i(\vhs_i-\vhhs_i) + \vm_i$, and \textcolor{black}{our policy in Lemma \ref{lemma:policy} is a subset of their policy.} Specifically, if $\hat{\Sigma}_{i}$ is invertible, the orthogonality constraint is redundant, and one can use the change of variable $\Gamma_i' = \Gamma_i \hat{\Sigma}_{i}^{-1}$ to show the equivalence of the policies. \textcolor{black}{However, in general, $\hat{\Sigma}_i$ may be singular,} and the orthogonality constraint is required for the convex optimization formulation in Lemma~\ref{lemma:n_letter_convex}. In the next lemma, the channel output is formalized as the measurement of a controlled state space.
\begin{lemma}[Channel outputs dynamics]\label{lemma:state-space}
For a fixed policy $\{(\Gamma_i,M_i)\}_{i=1}^n$, the channel outputs admit the state-space model
\begin{align}\label{eq:lemma_SS_SS}
    \vhs_{i+1}&= F\vhs_{i}  + K_{p,i} \ve_i, \ \ \ \nn\\
    \vy_i&=  (\Lambda \Gamma_i \hat{ \Sigma}_{i}^\dagger + H) \vhs_{i} - \Lambda \Gamma_i\hat{ \Sigma}_{i}^\dagger \vhhs_{i} + \Lambda \vm_i + \ve_i,
\end{align}
where $K_{p,i}$ and $\ve_i\sim N(0,\Psi_i)$ are defined in \eqref{eq:def_Kp_Psi}.
The estimator in \eqref{eq:vhhs_def} can be written as
\begin{align}\label{eq:vhhs_recursion}
    \vhhs_{i+1}&= F \vhhs_{i} + K_{Y,i}(\vy_i - H\vhhs_{i}),
\end{align}
and its error covariance $\hat{\Sigma}_{i}= \cov( \vhs_{i}-\vhhs_{i})$ satisfies the Riccati recursion
\begin{align}\label{eq:SIGMArecursion}
    \hat{ \Sigma}_{i+1}
    &= F \hat{ \Sigma}_{i}F^T + K_{p,i}\Psi_i K_{p,i}^T - K_{Y,i} \Psi_{Y,i} K_{Y,i}^T
\end{align}
with the initial condition $\hat{\Sigma}_1 =0$, and the constants
\begin{align}\label{eq:lemma_SS_KY}
     \Psi_{Y,i}&= (\Lambda \Gamma_i\hat{ \Sigma}_{i}^\dagger + H)\hat{ \Sigma}_{i}(\Lambda \Gamma_i\hat{ \Sigma}_{i}^\dagger + H)^T + \Lambda M_i \Lambda^T + \Psi_i \nn\\
    K_{Y,i} &=  (F\hat{ \Sigma}_{i}(\Lambda \Gamma_i\hat{ \Sigma}_{i}^\dagger + H)^T + K_{p,i}\Psi_i)\Psi_{Y,i}^{-1}.
    \end{align}
\end{lemma}
Lemma \ref{lemma:state-space} is a direct consequence of the policy derived in Lemma \ref{lemma:policy}. As seen from \eqref{eq:lemma_SS_SS}, the policy in \eqref{eq:policy} translates into an additive measurement noise $\vm_i$ and a modification of the observability matrix $\Lambda \Gamma_i \hat{ \Sigma}_{i}^\dagger + H$. Similar state-space structures appeared in \cite{Kim10_Feedback_capacity_stationary_Gaussian,charalambous2020new}, but it is interesting to realize that \eqref{eq:lemma_SS_SS} does not fall into the classical state-space structure since the observability matrix depends on the error covariance $\hat{\Sigma}_i$ induced from our policy. Lemma \ref{lemma:state-space} also reveals an objective structure that resembles the one in Theorem \ref{th:main}. \textcolor{black}{In particular, we can use the covariance of the channel outputs innovation in \eqref{eq:lemma_SS_KY}, $\Psi_{Y,i}$, and \eqref{eq:z_entropy_rate} to write}
\begin{align}
    h(\vy_i|\vy^{i-1}) - h(\vz_i|\vz^{i-1})&= \frac1{2}\log\det(\Psi_{Y,i}) - \frac1{2} \log\det(\Psi_{i}).\nn
\end{align}
The next lemma summarizes the SCOP formulation.
\begin{lemma}[Sequential convex-optimization formulation]\label{lemma:n_letter_convex}
The $n$-letter capacity can be bounded by the convex optimization problem
\begin{align}\label{eq:lemma_SCOP}
    &C_n(P)\le \max_{\{\Gamma_i,\Pi_i,\hat{ \Sigma}_{i+1}\}_{i=1}^n} \ \ \frac1{2n}\sum_{i=1}^n\log \det (\Psi_{Y,i}) - \log\det(\Psi_i)\nn\\
&\ \ \ \ \ s.t. \ \ \ \begin{pmatrix}
    \Pi_t & \Gamma_t\\
    \Gamma_t^T& \hat{ \Sigma}_t
\end{pmatrix} \succeq0, \ \ \frac1{n}\sum_{i=1}^n\mathbf{Tr}(\Pi_i)\le P,  \nn\\
&\Psi_{Y,t}=\Lambda \Pi_t\Lambda^T + H\hat{ \Sigma}_t H^T+ \Lambda \Gamma_t H^T+ H \Gamma_t^T\Lambda^T + \Psi_t \nn \\
& K_{Y,t}= (F \Gamma^T_t\Lambda^T + F \hat{ \Sigma}_{t}H^T + K_{p,t}\Psi_{t}) \Psi_{Y,t}^{-1}\nn \\
&\begin{pmatrix}
    F  \hat{ \Sigma}_t F^T + K_{p,t}\Psi_t K_{p,t}^T - \hat{ \Sigma}_{t+1}& K_{Y,t}\Psi_{Y,t} \\
    \Psi_{Y,t} K_{Y,t}^T & \Psi_{Y,t}
    \end{pmatrix}\succeq0, \textcolor{black}{\hat{\Sigma}_{n+1}\succeq0}
\end{align}
where the constraints hold for $t=1,\dots,n$, and $\hat{\Sigma}_1= 0$.
\end{lemma}
To see that \eqref{eq:lemma_SCOP} is a convex optimization, note that each of the matrix constraints is a linear function of the decision variables. In the next section, we provide the single-letter upper bound on the capacity. The key to the upper bound is the concavity of the objective function and the linearity of the constraints, along with the crucial property that the Riccati LMI constraint in \eqref{eq:lemma_SCOP} includes decision variables of two consecutive times only.

\subsection{Single-letter upper bound}
The next lemma concludes the upper bound in Theorem \ref{th:main}.
\begin{lemma}[The upper bound]\label{lemma:singleletter_UB}
The feedback capacity is bounded by the convex optimization
\begin{align}\label{eq:lemma_UB_OP}
&C_{fb}(P)\le \max_{\Pi,\hat{\Sigma},\Gamma} \frac1{2}\log \det (\Psi_Y) - \frac1{2}\log\det(\Psi)\nn\\
&\text{s.t. } \ \ \ \begin{pmatrix}
    \Pi & \Gamma\\
    \Gamma^T& \hat{ \Sigma}
\end{pmatrix} \succeq0, \ \ \mathbf{Tr}(\Pi)\le P, \ \nn\\
&\Psi_Y= \Lambda \Pi \Lambda^T + H \hat{ \Sigma} H^T+ \Lambda \Gamma H^T+ H \Gamma^T\Lambda^T  + \Psi\nn\\
&K_Y= (F \Gamma^T \Lambda^T + F \hat{ \Sigma} H^T + K_{p}\Psi) \Psi_Y^{-1}\nn\\
&\begin{pmatrix}
     F \hat{ \Sigma} F^T + K_{p}\Psi K_{p}^T - \hat{ \Sigma}& K_Y \Psi_Y \\
     \Psi_Y K_Y^T & \Psi_Y
    \end{pmatrix} \succeq0.
\end{align}
\end{lemma}
The main idea behind the upper bound is to show that the objective function evaluated at the convex combination of each of the decision variables in Lemma \ref{lemma:n_letter_convex} achieves a larger objective. At a high level, the idea is similar to the time-sharing random variable, but the challenge lies in the constraints. Specifically, one cannot show that the Riccati LMI constraint \eqref{eq:lemma_UB_OP} is satisfied at all times when evaluated at the convex combination of the decision variables. To settle this point, we show that this constraint is satisfied in the asymptotics.

\subsection{Lower bound}
In this section, we show that the upper bound in Lemma \ref{lemma:singleletter_UB} is achievable. It is shown with two lemmas: the first formulates a lower bound as an optimization problem that resembles the upper bound but has two additional constraints. The second lemma shows that additional constraints are satisfied in the upper bound optimization problem.
\begin{lemma}[Lower bound]\label{lemma:achievable}
For time-invariant policies
\begin{align}\label{eq:policy_LB}
\vx_i&= \Gamma (\vhs_{i}-\vhhs_{i}) + \vm_i,\ \ \ \ i\ge2
\end{align}
with $\vm_i\sim N(0,M)$ (and a different coding rule for $i=1$), the maximization of \eqref{eq:objectiveN} over $(\Gamma,M)$ achieves the lower bound
\begin{align}
    C_{fb}(P)&\geq \max_{\Gamma,\Pi,\hat{\Sigma}} \log\det(\Psi_Y) - \log\det(\Psi)\nn\\
    &\text{s.t.} \ \ \begin{pmatrix}
     \Pi &\Gamma\\\Gamma^T& \hat{\Sigma}
    \end{pmatrix} \succeq0, \ \ \mathbf{Tr}(\Pi)\le P \nn\\
    K_Y&= (F \hat{\Sigma}H^T + F{\Gamma}^T\Lambda^T + K_{p}\Psi)\Psi_{Y}^{-1}\nn\\
    \Psi_Y&= \Lambda\Pi\Lambda^T + H \hat{\Sigma}H^T+\Lambda \Gamma H^T + H \Gamma^T \Lambda^T + \Psi\nn\\
 \hat{ \Sigma} &= F \hat{ \Sigma} F^T + K_{p}\Psi K_{p}^T - K_Y \Psi_Y K_Y^T \label{eq:constraint_ricc}\\
 &\exists K: \rho(F-K(\Lambda \Gamma\hat{\Sigma}^\dagger + H))<1. \label{eq:constraint_detec}
\end{align}
\end{lemma}
The optimization problem in \eqref{eq:constraint_detec} is the same as the upper bound in \eqref{eq:lemma_UB_OP} except for the additional constraint \eqref{eq:constraint_detec} and the Riccati equation \eqref{eq:constraint_ricc} which appears as an inequality in the upper bound \eqref{eq:Schur_Riccati_main}. Next, we show that these two conditions are redundant concluding the proof of Theorem \ref{th:main}.

\begin{lemma}[Equality between the lower and upper bounds]\label{lemma:2conditions}
\textcolor{black}{For any} optimal tuple $(\Pi,\hat{\Sigma},\Gamma)$ for the upper bound optimization problem in \eqref{eq:lemma_UB_OP}:
\begin{enumerate}
\item \textcolor{black}{There exists an optimal tuple such that} the Schur complement of the Riccati LMI \eqref{eq:Schur_Riccati_main} is achieved with equality.
\item  The pair $(F,\Lambda \Gamma\hat{\Sigma}^\dagger + H)$ is detectable, i.e.,
$$\exists K: \rho(F-K(\Lambda \Gamma\hat{\Sigma}^\dagger + H))<1.$$
\end{enumerate}
Consequently, the upper bound in Lemma \ref{lemma:singleletter_UB} and the lower bound in Lemma \ref{lemma:achievable} are equal to the feedback capacity.
\end{lemma}
\textcolor{black}{For scalar channels with $H\neq0$, it can be shown that the optimal tuple satisfies the first item. That is, the Schur complement of the Riccati equation evaluated at any optimal solution tuple is zero. This fact is utilized in Theorem \ref{th:MA}.}
\section{Proof of Technical lemmas}\label{sec:proofs}
In this section, we provide detailed proofs of Lemmas \ref{lemma:policy} - \ref{lemma:2conditions} consecutively. We then prove Lemma \eqref{lemma:smoother} on the smoothing problem in Section \ref{sec:scheme}.

\begin{proof}[Proof of Lemma \ref{lemma:policy}]
\textcolor{black}{The policy in \eqref{eq:policy}} forms a subset of the maximization domain $P(\vx^n||\vy^n)= \prod_{i=1}^n P(\vx_i|\vx^{i-1},\vy^{i-1})$ \textcolor{black}{in \eqref{eq:objectiveN}}. Thus, our proof strategy is to construct a policy of the form \eqref{eq:policy}, for any inputs distribution $P(\vx^n||\vy^n)$, and show that it induces the same objective. The optimality of a Gaussian inputs distribution in \eqref{eq:objectiveN} can be shown with a standard argument of \textcolor{black}{maximum entropy}, e.g., \cite{Kim10_Feedback_capacity_stationary_Gaussian}.
We start by computing the $i$th objective as
\begin{align}\label{eq:proof_policy_obj}
    h(\vy_i|\vy^{i-1}) - h(\vz_i|\vz^{i-1})&= \frac1{2}\log \det (\mathbf{cov}(\vy_i-\hat{\hat{\vy_i}})) - \frac1{2}\log \det (\Psi_i),
\end{align}
where $\hat{\hat{\vy_i}} \triangleq \E[\vy_i|\vy^{i-1}]$. The covariance can be computed explicitly as
\begin{align}\label{eq:proof_policy_covP}
    \mathbf{cov}(\vy_i-\hat{\hat{\vy_i}})&\stackrel{(a)}= \mathbf{cov}(\vy_i-\hat{\hat{\vz_i}}) \nn\\
    &\stackrel{(b)}= \mathbf{cov}(\Lambda \vx_i + H\vhs_i - H\vhhs_i +  \vz_i - H\vhs_i)\nn\\
    &\stackrel{(c)}= \mathbf{cov}(\Lambda \vx_i + H(\vhs_i - \vhhs_i) ) + \Psi_i
\end{align}
\textcolor{black}{where $(a)$ follows from $\hat{\hat {\vz_i}} \triangleq \E[\vz_i|\vy^{i-1}]$ and $\E[\vx_i|\vy^{i-1}]=0$. The latter assumption is without loss of optimality since any policy with $\E[\vx_i|\vy^{i-1}]\neq0$ can be modified to \textcolor{black}{$\bar{\vx}_i = \vx_i - \E[\vx_i|\vy^{i-1}]$} that has zero mean without affecting the objective function in \eqref{eq:proof_policy_obj}. Step $(b)$ follows from the channel outputs definition in \eqref{eq:pre_channel} and \eqref{eq:vhhs_def}, and $(c)$ follows from the independence of the innovation $\vz_i-H\vhs_i$ and the tuple $(\vx^{i},\vy^{i-1},\vz^{i-1})$.}

\textcolor{black}{For any inputs distribution $P(\vx^n||\vy^n)$, denoted by $P$,} we construct a new policy of the form \eqref{eq:policy}, denoted by $Q$, as follows
\begin{align}\label{eq:proof_newpolicy}
  \vx_i = \Gamma_i \hat{\Sigma}_{i}^\dagger(\vhs_i - \vhhs_i) + \mathbf{m}_i,
\end{align}
where $\Gamma_i\triangleq\E_P[\vx_i(\vhs_i-\vhhs_i)^T ]$, $\mathbf{m}_i$ is independent of $(\vx^{i-1},\vy^{i-1})$ and is distributed according to $\mathbf{m}_i\sim N(0,M_i)$ with
\begin{align}\label{eq:proof_policy_M}
    M_i\triangleq \E_P[ \vx_i \vx_i^T] - \E_P[ \vx_i(\vhs_i-\vhhs_i)^T]\hat{\Sigma}_{i}^\dagger \E_P[(\vhs_i-\vhhs_i)\vx_i^T],
\end{align}
and $\hat{\Sigma}_{i}^\dagger$ is the pseudo inverse of $\hat{\Sigma}_{i} \triangleq \mathbf{cov}_P(\vhs_i-\vhhs_i)$. The subscript $P$ is made to emphasize the dependence on the distribution~$P$.

We show by induction that the new policy in \eqref{eq:proof_newpolicy} induces the same objective as the distribution $P$. Consider the Gaussian vector $\Xi^{P/Q}_i \triangleq (\vhs_i-\vhhs_i,\vx_i,\vy_i-\hat{\hat{\vy_i}})$ where the superscript indicates its distribution. If we show that~$\Xi^P_i$ has the same distribution as~$\Xi^Q_i$ for $i=1,\dots,n$, then their objectives are equal by~\eqref{eq:proof_policy_obj}. For the base case of the induction, we have $\Xi_1^{P/Q} = (0,0,\vx_1,\vy_1)$ for both policies and our construction in \eqref{eq:proof_newpolicy} guarantees that $\vx_1$ has the same distribution for both policies. For the induction step, assume that the variables $\{\Xi^P_i\}_{i\le t}$ have the same distribution as $\{\Xi^Q_i\}_{i\le t}$. We show that tuple $\Xi_{t+1}^P$ has the same distribution as $\Xi_{t+1}^Q$ by comparing their different components using a Bayes rule. First, the encoders' estimate $\vhs_{t+1}$ is independent of the policy choice. The decoders' estimate $\vhhs_{t+1} = \E[\vhs_{t+1}|\vy^t]$ is a function of the innovations $\{\vy_i-\hat{\hat{\vy}}_i\}_{i\le t}$, and by the induction hypothesis these innovations have the same distribution. These first two steps conclude that $\mathbf{cov}_P(\vhs_{t+1}-\vhhs_{t+1}) = \mathbf{cov}_Q(\vhs_{t+1}-\vhhs_{t+1})$. For the channel input, it can be easily verified by \eqref{eq:proof_newpolicy} that $\E_Q[\vx_{t+1}\vx_{t+1}^T] = \E_P[\vx_{t+1}\vx_{t+1}^T]$, and below we also show that $$\E_Q[\vx_{t+1}(\vhs_{t+1}-\vhhs_{t+1})^T] = \Gamma_{t+1} \hat{\Sigma}_{t+1}^\dagger \hat{\Sigma}_{t+1} = \Gamma_{t+1} = \E_P[\vx_{t+1}(\vhs_{t+1}-\vhhs_{t+1})^T].$$ The last step to complete the inductive step is for the innovation $\vy_{t+1}-\hat{\hat{\vy}}_{t+1}$, and we note from \eqref{eq:proof_policy_covP} that the distribution of the latter conditioned on $(\vhs_{t+1}-\vhhs_{t+1})$, $\vx_{t+1}$ is determined by $\vz_{t+1}-H\vhs_{t+1}$.

The orthogonality constraint $\Gamma_i(I-\hat{\Sigma}_i^\dagger\hat{\Sigma}_i)$ is a property of covariance matrices since $(I-\hat{\Sigma}^\dagger_{i}\hat{\Sigma}_{i})$ is the orthogonal projection onto the kernel of $\hat{\Sigma}_{i}$, but we prove it here for completeness. Consider the eigendecomposition of the covariance matrix
\begin{align}\label{eq:proof_policy_decomp}
    \hat{\Sigma}_{i}&= \E[(\vhs_i-\vhhs_i)(\vhs_i-\vhhs_i)^T]\nn\\
    &= \begin{pmatrix} U_0&U_1\end{pmatrix} \begin{pmatrix} \Omega& 0\\ 0&0\end{pmatrix} \begin{pmatrix} U_0^T\\U_1^T\end{pmatrix},
\end{align}
where $\begin{pmatrix} U_0&U_1\end{pmatrix}$ is an orthogonal matrix and $\Omega\succ0$ which imply $(\vhs_i-\vhhs_i)^TU_1 = 0$. The Moore-Penrose pseudo inverse is
\begin{align}
    \hat{\Sigma}_{i}^\dagger&=  \begin{pmatrix} U_0&U_1\end{pmatrix} \begin{pmatrix} \Omega^{-1}& 0\\ 0&0\end{pmatrix} \begin{pmatrix} U_0^T\\U_1^T\end{pmatrix},
\end{align}
and the constraint can be written as
\begin{align}\label{eq:proof_policy_orthogo}
    &\E[\vx_i(\vhs_i - \vhhs_i)^T] (I-\hat{\Sigma}^\dagger_{i}\hat{\Sigma}_{i})\nn\\
    &\ = \E[\vx_i(\vhs_i - \vhhs_i)^T] \left(I- \begin{pmatrix} U_0&U_1\end{pmatrix} \begin{pmatrix} I & 0\\ 0&0\end{pmatrix} \begin{pmatrix} U_0^T\\U_1^T\end{pmatrix}\right).
\end{align}
\textcolor{black}{To see that \eqref{eq:proof_policy_orthogo} is the zero matrix, note that if $u$ is a column of $U_0$, then $\left(I- \begin{pmatrix} U_0&U_1\end{pmatrix} \begin{pmatrix} I & 0\\ 0&0\end{pmatrix} \begin{pmatrix} U_0^T\\U_1^T\end{pmatrix}\right)u=0$. Further, if $u$ is a column of $U_1$, then $\left(I- \begin{pmatrix} U_0&U_1\end{pmatrix} \begin{pmatrix} I & 0\\ 0&0\end{pmatrix} \begin{pmatrix} U_0^T\\U_1^T\end{pmatrix}\right)u=u$, but $(\vhs_i - \vhhs_i)^Tu=0$ by the decomposition in \eqref{eq:proof_policy_decomp}.}

Finally, it can be verified that the power consumed by the new policy satisfies
\begin{align}
 \sum_{i=1}^n \E_Q[\vx_i^T\vx_i]
= \sum_{i=1}^n \E_P[\vx_i^T\vx_i].\nn
\end{align}
\end{proof}

\begin{proof}[Proof of Lemma \ref{lemma:state-space}]
The recursion for the predicted state $\vhs_{i+1}$ is given in Eq. \eqref{eq:prel_estimator} where $\mathbf{e}_i$ is the innovation process. For the channel output, we use Lemma \ref{lemma:policy} to write
\begin{align}
    \vy_i &= \Lambda \vx_i + \vz_i\nn\\
        &= (\Lambda \Gamma_i \hat{ \Sigma}_{i}^\dagger+ H) \vhs_{i} - \Lambda \Gamma_i \vhhs_{i} + \Lambda \mathbf{m}_i + \mathbf{e}_i.
\end{align}
Note that the term $\vhhs_{i}$ is a deterministic function of $\vy^{i-1}$ and thus has no effect on the estimation error. To show that \eqref{eq:lemma_SS_SS} is a state-space model that admits standard Kalman filtering, note that the measurement noise $\Lambda \mathbf{m}_i + \ve_i$ is independent of $\vz^{i-1}$. Thus, the measurement noise is independent of previous measurements $\vy^{i-1}$ and the hidden states $\vhs^{i-1}$ of the state-space model.

To obtain the optimal estimator and the error covariance recursion in \eqref{eq:SIGMArecursion}, we apply the standard Kalman filter recursions \eqref{eq:prel_estimator}-\eqref{eq:Riccati_recursion} that also hold with the time-varying constants $G = K_{p,i}$, $H = \Lambda\Gamma_i\hat{ \Sigma}_{i}^\dagger + H$, $W = S = \Psi_i$ and $V= \Lambda M_i \Lambda^T + \Psi_i$.
\end{proof}
\begin{proof}[Proof of Lemma \ref{lemma:n_letter_convex}]
The starting point is the combination of Lemma \ref{lemma:policy} and Lemma \ref{lemma:state-space} to the optimization problem of $C_n(P)$
\begin{align}
&\max \frac1{2n}\sum_{i=1}^n \log \det (\Psi_{Y,i}) - \log\det(\Psi_i)\nn\\
& \ \ \ \ \text{s.t.} \ \ \frac1{n}\sum_{i=1}^n\mathbf{Tr}( \Gamma_i \hat{\Sigma}^\dagger_{i}\Gamma_i^T + M_i)\le P, \nn\\
& \ \ \Gamma_i (I - \hat{\Sigma}^\dagger_{i}\hat{\Sigma}_{i}) =0,\ \ M_i\succeq 0 \nn\\
& \Psi_{Y,i} = (\Lambda\Gamma_i\hat{ \Sigma}_{i}^\dagger + H ) \hat{ \Sigma}_{i} (\Lambda\Gamma_i\hat{ \Sigma}_{i}^\dagger + H )^T + \Lambda M_i \Lambda^T + \Psi_i\nn\\
    &K_{Y,i} =  (F\hat{ \Sigma}_{i}(\Lambda \Gamma_i\hat{ \Sigma}_{i}^\dagger + H)^T + K_{p,i}\Psi_i)\Psi_{Y,i}^{-1}\nn\\
&\hat{ \Sigma}_{i+1}
    = F \hat{ \Sigma}_{i}F^T + K_{p,i}\Psi_i K_{p,i}^T - K_{Y,i} \Psi_{Y,i} K_{Y,i}^T
\end{align}
with the initial condition $\hat{ \Sigma}_{1} = 0$. The maximum is over all involved variables, that is, $\{\Gamma_i,M_i,\hat{ \Sigma}_{i+1}\}_{i=1}^n$.

The first step is to introduce an auxiliary decision variable
\begin{align}\label{eq:defPi}
    \Pi_i&= \Gamma_i \hat{ \Sigma}^\dagger_{i} {\Gamma}_i^T + M_i.
\end{align}
Then, $\Psi_{Y,i}$ can be written as
\begin{align}\label{eq:proof_obj_Pi}
\Psi_{Y,i}&=\Lambda \Pi_i\Lambda^T + H\hat{ \Sigma}_{i} H^T+ \Lambda \Gamma_i H^T+ H \Gamma_i^T\Lambda^T  + \Psi_i,\nn
\end{align}
where we used the orthogonality constraint $\Gamma_i(I - \hat{\Sigma}^\dagger_{i}\hat{\Sigma}_{i}) =0$. As a result, the Riccati recursion can also be represented with $\Pi_i$ only. The power constraint can also be expressed as
\begin{align}
    \frac1{n}\sum_{i=1}^n\mathbf{Tr}(\Pi_i)\le P,
\end{align}
so that the variable $M_i$ only appears in the constraints
\begin{align}
    \Pi_i&= \Gamma_i \hat{ \Sigma}^\dagger_i \Gamma_i^T + M_i \nn\\
    M_i&\succeq0,
\end{align}
which can be reduced to the constraint
\begin{align}
    \Pi_i&\succeq \Gamma_i \hat{ \Sigma}^\dagger_{i} \Gamma_i^T.
\end{align}
By the Schur complement for positive semidefinite matrices \cite[p. $651$]{BoydOptimizationBook04},
\begin{align}
        \hat{\Sigma}_i\succeq0 \ \& \ \Pi_i- \Gamma_i \hat{ \Sigma}^\dagger_{i} \Gamma_i^T\succeq0 \ \&\ \Gamma_i(I - \hat{\Sigma}^\dagger_{i}\hat{\Sigma}_{i}) =0  \iff \begin{pmatrix}
    \Pi_i & \Gamma_i\\
    \Gamma_i^T& \hat{ \Sigma}_i
\end{pmatrix} \succeq0.
\end{align}

Finally, the Riccati equation is relaxed to the Riccati inequality
\begin{align}
\hat{ \Sigma}_{i+1}
    &\preceq F \hat{ \Sigma}_{i}F^T + K_{p,i}\Psi_i K_{p,i}^T - K_{Y,i} \Psi_{Y,i} K_{Y,i}^T,
\end{align}
and using the Schur complement transformation, we can write
\begin{align}\label{eq:proof_SCOP_riccati_i}
    \begin{pmatrix}
    F  \hat{ \Sigma}_{i} F^T + K_{p,i}\Psi_i K_{p,i}^T- \hat{ \Sigma}_{i+1}& K_{Y,i} \Psi_{Y,i} \\
    \Psi_{Y,i} K_{Y,i}^T & \Psi_{Y,i}
    \end{pmatrix}\succeq0.
\end{align}

\end{proof}

\begin{proof}[Proof of Lemma \ref{lemma:singleletter_UB}]
This is the converse proof for the capacity expression in Theorem \ref{th:main}. Recall that throughout the derivations, we used the $n$-letter capacity $C_n(P)$ in \eqref{eq:objectiveN}, but a standard converse argument can relate this quantity to the feedback capacity by showing that for any $n$,
\begin{align}\label{eq:proof_UB_fano}
    C_{fb}(P)&\le \frac1{n}C_n(P) + \delta_n
\end{align}
where $\delta_n\to 0$ is resulted from a Fano's inequality. The remaining step is to show that the SCOP formulation in Lemma \ref{lemma:n_letter_convex} that serves as an upper bound to $C_n(P)$ can be further upper bounded by its single-letter counterpart, the optimization problem in Theorem \ref{th:main}.

Define the convex combinations of the decision variables as
\begin{align}\label{eq:proof_UB_convexcomb}
    \bar{\Pi}_n&= \frac1{n}\sum_{i=1}^n \Pi_i,
    \ \ \ \bar{\Gamma}_n= \frac1{n}\sum_{i=1}^n \Gamma_i, \ \ \ \bar{\hat {\Sigma}}_n= \frac1{n}\sum_{i=1}^n \hat {\Sigma}_i,
\end{align}
and also let $\bar{\Sigma}_n\triangleq \frac1{n}\sum_{i=1}^n \Sigma_{i}$, $\bar{\Psi}_{n}\triangleq \frac1{n}\sum_{i=1}^n \Psi_i$ denote the averaged constants of the Riccati variables.

The concavity of the $\log\det(\cdot)$ function and Jensen's inequality imply that the convex combinations attain a greater objective than the one in Lemma \ref{lemma:n_letter_convex},
\begin{align}
    \frac1{n}\sum_{i=1}^n \log \det (\Psi_{Y,i})&\le \log \det \left(\frac1{n}\sum_{i=1}^n \Psi_{Y,i}\right),
\end{align}
where the argument of the right-hand side can be written as the linear function
\begin{align}
    &\frac1{n}\sum_{i=1}^n \Psi_{Y,i} \nn\\
    &\ = \Lambda \bar{\Pi}_n\Lambda^T + H\bar{\hat{ \Sigma}}_n H^T+ \Lambda \bar{\Gamma}_n H^T+ H \bar{\Gamma}_n^T\Lambda^T  + H \bar{\Sigma}_nH^T.    \nn
\end{align}
Next, the per-time constraints of the $n$-letter problem should be transformed into their single-letter counterparts, that is, the ones evaluated at the convex combinations in \eqref{eq:proof_UB_convexcomb}. It is straightforward to show that the power constraint and the first LMI constraint are satisfied at the convex combination by
\begin{align}
    \text{Tr}(\bar{\Pi}_n)&= \frac1{n}\sum_{i=1}^n\text{Tr}(\Pi_i)\nn\\
    &\le P,
\end{align}
and
\begin{align}
\begin{pmatrix}
    \bar{\Pi}_n & \bar{\Gamma}_n\\
    \bar{\Gamma}_n^T& \bar{\hat{ \Sigma}}_n
\end{pmatrix}
&=\frac1{n}\sum_{i=1}^n \begin{pmatrix}
    \Pi_i & \Gamma_i\\
    {\Gamma}_i^T& \hat{ \Sigma}_{i}
\end{pmatrix}\nn\\ &\succeq 0.
\end{align}
We proceed to the last constraint in the optimization problem, the Riccati LMI, defined by
\begin{align}\label{eq:proof_UB_Ricc}
\Omega(\Pi,\hat{\Sigma},\Gamma)&= \begin{pmatrix}
    F  \hat{ \Sigma} F^T - \hat{ \Sigma} + K_{p}\Psi K_{p}^T & K_Y(\Gamma,\hat{\Sigma}) \\
    K_Y(\Gamma,\hat{\Sigma})^T & \Psi_Y(\Gamma,\hat{\Sigma},\Pi)
    \end{pmatrix}
\end{align}
with
\begin{align*}
    K_Y(\Gamma,\hat{\Sigma})&=F {\Gamma}^T\Lambda^T + F \hat{ \Sigma} H^T + K_{p}\Psi\nn\\
    \Psi_Y(\Gamma,\hat{\Sigma},\Pi)&= \Lambda \Pi\Lambda^T + H\hat{ \Sigma} H^T+ \Lambda \Gamma H^T+ H  {\Gamma}^T\Lambda^T + \Psi.
\end{align*}
The main challenge is that the Riccati LMI does not satisfy $\Omega(\bar{\Pi}_n,\bar{\Sigma}_n,\bar{\Gamma}_n)\succeq0$ for all $n$. \textcolor{black}{In other words, the tuple of convex combinations in \eqref{eq:proof_UB_convexcomb} does not lie in the constraint set of the convex optimization in Theorem \ref{th:main}. Our strategy is to show that the limiting tuple of convex combinations (as a function of $n$) lies in the required constraint set. This is achieved by showing that the tuple of convex combinations lies in a relaxed constraints set, parameterized with some $\epsilon>0$. We then show that $\epsilon$ can be made small as $n$ grows large and argue that there is a limit point that lies in the constraints set that corresponds to $\epsilon=0$.}

Define the $\epsilon$-domain of the constraints set as
\begin{align*}
    \mathcal{C}_{\epsilon}&= \{\begin{pmatrix}
    \Pi & \Gamma\\
    {\Gamma}^T& \hat{\Sigma}
\end{pmatrix}\succeq0: \Omega(\Pi,\hat{\Sigma},\Gamma) + \epsilon I\succeq 0, \mathbf{Tr}(\Pi)\le P\},
\end{align*}
and note that $\mathcal C_0$ is the constraints set in Theorem \ref{th:main}.

By summing over both sides of the Riccati inequality in \eqref{eq:proof_SCOP_riccati_i}, we have
\begin{align*}
    &\begin{pmatrix}
    \frac1{n}\sum_{i=1}^{n}\hat{ \Sigma}_{i+1}& 0 \\
    0 & 0
    \end{pmatrix}\preceq \frac1{n}\sum_{i=1}^{n} \begin{pmatrix}
     F  \hat{ \Sigma}_{i} F^T + K_{p,i}\Psi_i K_{p,i}^T & K_{Y,i}\Psi_{Y,i} \\
    \Psi_{Y,i} K_{Y,i}^T & \Psi_{Y,i},
    \end{pmatrix}
\end{align*}
Arranging both sides and using the fact that $\Sigma_{n+1}\succeq0$,
\begin{align}
    \begin{pmatrix}
    \bar{\hat{ \Sigma}}_{n} & 0 \\
    0 & 0
    \end{pmatrix}
    &\preceq \begin{pmatrix}
     F \bar{\hat{ \Sigma}}_{n} F^T + K_{p}\Psi K_{p}^T &  K_Y(\bar{\Gamma}_n,\bar{\hat{ \Sigma}}_n) \\
    K_Y(\bar{\Gamma}_n,\bar{\hat{ \Sigma}}_n)^T & \Psi_Y(\bar{\Gamma}_n,\bar{\hat{ \Sigma}}_n,\bar{\Pi}_n)
    \end{pmatrix}\nn\\
    & + \begin{pmatrix}
     \bar{\Psi}_n -\Psi & F (\bar{\Sigma}_n -\Sigma ) H^T \\
    H (\bar{\Sigma}_n -\Sigma ) F^T & \bar{\Psi}_n - \Psi
    \end{pmatrix}.\nn
\end{align}
By our assumptions on the state-space model of the noise, we can use \cite[Ch. $14$]{kailath_booklinear} to have $\bar{\Sigma}_n\to \Sigma$ and $\bar{\Psi}_n \to \Psi$. Thus, the constraint on $\Omega (\bar{\Pi}_n,\bar{\hat{ \Sigma}}_{n},\bar{\Gamma}_n)$ is satisfied \emph{asymptotically}. Specifically, for any $\epsilon$, there exists an $n_\epsilon$ such that for all $n>n_\epsilon$
\begin{align}
    0 \preceq \Omega(\bar{\Pi}_n,\bar{\hat{ \Sigma}}_{n},\bar{\Gamma}_n) + \epsilon I.
\end{align}
Since the set $\mathcal{C}_\epsilon$ is closed and nested (in $\epsilon)$, the sequence $\{(\bar{\Pi}_{n},\bar{\hat{ \Sigma}}_{n},\bar{\Gamma}_{n}) \}_{n\in\mathbb{N}}$ has a limit point in $\bigcap_{\epsilon>0} \mathcal{C}_\epsilon = \mathcal C_0$. That is, there exists a sequence of times $T_1\le T_2\le T_3\dots$ such that $\lim_{i\to\infty} (\bar{\Pi}_{T_i},\bar{\hat{ \Sigma}}_{T_i},\bar{\Gamma}_{T_i}) \in \mathcal C_0$. It is important to note that the times sequence depends on the noise characteristics and not on the underlying codebooks. The proof is completed by taking the limit over the sequence $T_1,T_2,\dots$ in \eqref{eq:proof_UB_fano} to obtain
\begin{align}
    C_{fb}(P) &\le \max_{(\Pi,\hat{\Sigma},\Gamma) \in \mathcal C_0} \frac1{2}\log \det (\Psi_Y(\Gamma,\hat{\Sigma},\Pi)) - \frac1{2}\log\det(\Psi),\nn
\end{align}
which is precisely the optimization problem in \eqref{eq:lemma_UB_OP}.

\end{proof}
\begin{proof}[Proof of Lemma \ref{lemma:achievable}]
This is the achievability proof of the optimization problem in Lemma \ref{lemma:achievable}. The main ides is to fix a time-invariant policy and analyze the achievable rate which is determined by the asymptotic behaviour of the channel outputs process. Since the channel outputs process is described as a state-space, the asymptotic behavior of the channel outputs statistics boils down to the analysis of Riccati recursion convergence. To that end, we will use a result from \cite{convergence_detectable_initial} on certain conditions to guarantee the convergence of the Riccati recursion to the Riccati equation. Lastly, since one of the condition is given on the initial condition of the Riccati recursion (which we have no direct control over), we modify the time-invariant policy at the first time only to guarantee the convergence.

We use the policy in Lemma \ref{lemma:policy} with $\Gamma_i = \Gamma \hat{\Sigma}_{i}$ and $M_i= M$
such that the corresponding power satisfies $$\frac1{n}\sum_{i=1}^n \text{Tr}(\Gamma \hat{\Sigma}_{i}\Gamma^T + M)\le P.$$ By Lemma \ref{lemma:state-space}, the induced state-space is
\begin{align}
    \vhs_{i+1}&= F\vhs_{i}  + K_{p,i} \ve_i, \ \ \ \nn\\
    \vy_i&=  (\Lambda \Gamma + H) \vhs_{i} - \Lambda \Gamma \vhhs_{i} + \Lambda \mathbf{m}_i + \ve_i,
\end{align}
and the corresponding Riccati recursion is
\begin{align}\label{eq:ach_recursion}
    \hat{ \Sigma}_{i+1}
    &= F \hat{ \Sigma}_{i}F^T + K_{p,i}\Psi_i K_{p,i}^T - K_{L,i} \Psi_{L,i} K_{L,i}^T
\end{align}
with $\hat{\Sigma}_1=0$ and
\begin{align}
     \Psi_{L,i}&= (\Lambda \Gamma + H)\hat{ \Sigma}_{i}(\Lambda \Gamma + H)^T + \Lambda M \Lambda^T + \Psi_i \nn\\
    K_{L,i} &=  (F\hat{ \Sigma}_{i}(\Lambda \Gamma  + H)^T + K_{p,i}\Psi_i)\Psi_{L,i}^{-1}.
\end{align}
The next step is to show the convergence of the Riccati recursion in \eqref{eq:ach_recursion} to a fixed-point solution of the Riccati equation. Since $K_{p,i}$ and $\Psi_i$ converge to their time-invariant counterparts in \eqref{eq:ricc_constants} exponentially fast, we replace $K_{p,i}$ and $\Psi_i$ with $K_p$ and $\Psi$, respectively. This comes at the cost that the initial condition $\hat{\Sigma}_1=0$ becomes arbitrary.

Before presenting the convergence conditions, we need to modify the Riccati recursion in \eqref{eq:ach_recursion} to have an equivalent form with the property that the disturbance and the measurement of the state are independent. This is a standard modification can be found for instance in \cite[Sec. 14.7]{kailath_booklinear}. The equivalent form of \eqref{eq:ach_recursion} can be written as
\begin{align}\label{eq:_proof_LB_recur_noCOR}
    \hat{ \Sigma}_{i+1}
    &= F_s \hat{ \Sigma}_{i}F_s^T + K_{p}Q_s K_{p}^T - \bar{K}_{L,i} \Psi_{L,i} \bar{K}_{L,i}^T,
\end{align}
and
\begin{align}
F_s&= F - K_p \Psi (\Lambda M \Lambda^T + \Psi)^{-1}(\Lambda \Gamma + H)\nn\\
Q_s&= \Psi - \Psi (\Lambda M \Lambda^T + \Psi )^{-1}\Psi\nn\\
\bar{K}_{L,i}&=  F_s \hat{\Sigma}_i (\Lambda \Gamma + H)^T \Psi_{L,i}^{-1}.
\end{align}

We use \cite[Th. $1$]{convergence_detectable_initial} for the convergence of the Riccati recursion in \eqref{eq:_proof_LB_recur_noCOR} to the maximal solution of the Riccati equation, the maximal solution $\hat{\Sigma}_s$ whose all of its closed-loop modes are inside or on the unit circle, that is, $\rho(F_s - \bar{K}_{L,i} (\Lambda \Gamma + H))\le 1$. The sufficient condition from \cite{convergence_detectable_initial} translates to the Riccati equation in \eqref{eq:_proof_LB_recur_noCOR} as
\begin{enumerate}
    \item The initial state satisfies $\hat{\Sigma}_1\succeq \hat{\Sigma}_s$.
    \item The pair $(F_s,\Lambda \Gamma + H)$ is detectable.
\end{enumerate}
The detectability condition guarantees the existence of the maximal solution. This condition will be carried to the lower bound optimization problem as a restriction on the optimization parameters $(\Gamma,M)$. Also note that $(F_s,\Lambda\Gamma +H)$ is detectable iff $(F,\Lambda\Gamma +H)$ is detectable and thus can be expressed as $\exists K: \rho(F-K(\Lambda \Gamma + H))<1$. The first condition is needed for the convergence to the maximal solution and is shown next.
As mentioned, the initial condition $\hat{\Sigma}_1$ is arbitrary. To this end, we modify the time-invariant policy by changing $M_1$ to be an identity matrix scaled with a constant $\alpha$.

We proceed to show that the null-space of $\hat{\Sigma}_2$ lies in the null-space of any solution to the Riccati equation. For this proof, we use the closed-loop Lyapunov recursion of \eqref{eq:_proof_LB_recur_noCOR} can be expressed as
\begin{align}\label{eq:proof_LB_lyapunov}
\hat{\Sigma}_2&= (F_s - \bar{K}_{L,1}(\Lambda\Gamma + H))\hat{\Sigma}_1(F_s - \bar{K}_{L,1}(\Lambda\Gamma + H))^T \nn\\
 &\ \ + K_pQ_sK_p^T + \bar{K}_{L,1}(\Lambda M_1 \Lambda^T + \Psi) \bar{K}_{L,1}^T.
\end{align}
Let $x$ be an eigenvector of $F$ with $\lambda$ such that $x\hat{\Sigma}_2=0$. Then, pre- and post- multiplying the closed-loop Riccati equation in \eqref{eq:_proof_LB_recur_noCOR} with $x$ and $x^T$ we have
\begin{align}
 0&= x (F_s - \bar{K}_{L,1}(\Lambda\Gamma + H))\hat{\Sigma}_1(F_s - \bar{K}_{L,1}(\Lambda\Gamma + H))^Tx^T \nn\\
 &\ \ + x K_pQ_sK_p^Tx^T + x \bar{K}_{L,1}(\Lambda M_1 \Lambda^T + \Psi) \bar{K}_{L,1}^Tx^T.
\end{align}
Then, we have $xK_pQ_s=0$, $x\bar{K}_{L,1}=0$ which also implies $xF_s\hat{\Sigma}=0$. By $M_1\succ0$, we have $Q_s\succ0$ so that $xK_p=0$. Now, consider any solution to the Riccati equation. Then, pre- and post- multiplying the Riccati equation with $x$ and $x^T$ gives
$$x\hat{\Sigma}x^T = x F_s\hat{\Sigma}F_s^Tx^T + xK_pQ_sK_p^Tx^T - x^T\bar{K}_{L}\Psi_L\bar{K}_L^Tx^T,$$
which implies $x\hat{\Sigma}x^T(1-|\lambda|^2)\preceq0$. Finally, by the stability of $F-K_p H$, the equation $xK_p=0$ implies $|\lambda|<1$ and therefore $x\hat{\Sigma}=0$. To conclude the proof of the first item, we can choose $\alpha$ to be large enough such that the error covariance $\hat{\Sigma}_2\succeq \Sigma_s$. Note that the power constraint may be violated for small $n$ but it will average out when taking $n$ to be large enough.

To summarize, for any time-invariant policy $(M,\Gamma)$ subject to the detectability condition, the channel outputs entropy rate converges to
\begin{align}
    \lim_{n\to\infty} \frac1{n}h(Y^n) &= \frac1{2}\log\det(\Psi_{Y,s}) + \frac1{2}\log(2\pi e)^d
\end{align}
where $\Psi_{Y,s}$ is the innovation covariance of the Riccati equation in \eqref{eq:_proof_LB_recur_noCOR} evaluated at its (unique) maximal solution.




As shown in \cite{CoverPombra}, the asymptotic equipartition property (AEP) holds for arbitrary Gaussian processes, so that $\lim_{n\to\infty} \frac1{n}(h(Y^n) - h(Z^n))$ is achievable for any policy of the form
$X^n = B_nZ^n + V^n$ where $V^n\sim(0,\Sigma_{V_n})$ is independent of $Z^n$ and $B_n$ is a (block) lower-triangular matrix, i.e., it is a strictly causal operator. The policy considered here can be written in this form since $\vhs_i$ is a strictly causal function of $\{\vz_i\}_{i\geq1}$ and $\vhhs_i$ is a strictly causal function of $\{\vy_i\}_{i\geq1}$. Thus, we have that
\begin{align}
    C_{fb}(P) &\ge \frac1{2} \log\det(\Psi_{Y,s}) - \frac1{2} \log\det(\Psi).
\end{align}

We formulate an optimization problem which serves as a lower bound on the feedback capacity. By taking a maximum over all valid policies, we have
\begin{align}\label{eq:lemma_lb_OP}
    C_{fb}(P)&\geq \max_{\Gamma,M,\hat{\Sigma}_s} \frac1{2}\log\det(\Psi_{Y,s}) - \frac1{2}\log\det(\Psi)\nn\\
    \text{s.t.} \ \ & \text{Tr}(\Gamma \hat{\Sigma}_s\Gamma^T + M )\le P \nn\\
     \hat{ \Sigma}_s &= F \hat{ \Sigma}_s F^T + K_{p}\Psi K_{p}^T - K_L \Psi_L K_L^T\nn\\
    K_L&= (F\hat{\Sigma}_s(\Lambda \Gamma  + H)^T + K_{p}\Psi)\Psi_{L}^{-1}\nn\\
    \Psi_{Y,s}&= (\Lambda \Gamma + H)\hat{ \Sigma}_s(\Lambda \Gamma + H)^T + \Lambda M \Lambda^T + \Psi \nn\\
    & \exists K: \rho(F-K(\Lambda \Gamma + H))<1,
\end{align}
To complete the proof, change the variable $\Gamma' = \Gamma \hat{\Sigma}_s$, add the orthogonality constraint and follow the steps in Lemma \ref{lemma:singleletter_UB}: define $\Pi = \Gamma \hat{\Sigma}_s^\dagger\Gamma^T + M$, reduce $M$ and apply the Schur complement to get the optimization problem \eqref{eq:lemma_lb_OP}. For consistency with the upper bound notation, we rename $\Gamma'$ and $\hat{\Sigma}_s$ with $\Gamma$ and $\hat{\Sigma}$ respectively.

\end{proof}
\begin{proof}[Proof of Lemma \ref{lemma:2conditions}]
Recall that from the upper bound optimization problem, the tuple $(\Pi,\hat{\Sigma},\Gamma)$ satisfies
\begin{align}\label{eq:proof_ach_Ricc1}
    \hat{\Sigma}
    &\preceq F \hat{\Sigma}F^T + K_{p}\Psi K_{p}^T - K_{Y} \Psi_{Y} K_{Y}^T,
\end{align}
with
\begin{align}\label{eq:proof_lemma_conditions_obj}
     \Psi_{Y}&= \Lambda  \Pi  \Lambda^T + H \hat{\Sigma} H^T + \Lambda\Gamma H^T + H \Gamma^T\Lambda^T + \Psi \nn\\
     &= (\Lambda \Gamma\hat{\Sigma}^\dagger + H)\hat{\Sigma}(\Lambda \Gamma\hat{\Sigma}^\dagger + H)^T + \Lambda (\Pi - \Gamma\hat{\Sigma}^\dagger\Gamma^T) \Lambda^T + \Psi \nn\\
    K_{Y} &= (F\hat{\Sigma}(\Lambda \Gamma\hat{\Sigma}^\dagger + H)^T + K_{p}\Psi)\Psi_{Y}^{-1}.
\end{align}
We prove the claims.
\begin{enumerate}
\item If the optimal tuple does not satisfy the Riccati inequality \eqref{eq:proof_ach_Ricc1} with equality, there exists a matrix $Q\succeq0$ such that
\begin{align}
    Q \triangleq F \hat{\Sigma}F^T - \hat{\Sigma} + K_{p}\Psi K_{p}^T - K_{Y} \Psi_{Y} K_{Y}^T
\end{align}
is not the zero matrix. We let $\hat{\Sigma}' = Q + \hat{\Sigma}$, and observe that this modification satisfies the power constraint, and the LMI $$\begin{pmatrix}
    \Pi & \Gamma\\
    \Gamma^T& \hat{\Sigma}'
\end{pmatrix} \succeq0.$$ Then, using $\hat{\Sigma}' \succeq \hat{\Sigma}$ and the optimality of the tuple $(\Gamma,\Pi,\hat{\Sigma})$, we conclude that the objective is still equal to its optimal value after the modification.
\item If there exists an unstable mode in $F$ that cannot be observed via $\Lambda\Gamma\Sigma^\dagger + H$, by our assumption that $(F,H)$ is detectable, this mode can be observed via $\Lambda\Gamma\hat{\Sigma}^\dagger$. On the other hand, the instability of this mode implies that the error covariance $\hat{\Sigma}$ has an infinite value in this direction which is a contradiction to the observability of this mode via the matrix $\Lambda\Gamma\hat{\Sigma}^\dagger$.
\end{enumerate}

\end{proof}

\subsection{Proof for the coding scheme analysis} \label{subsec:proofs_coding}
\begin{proof}[Proof of Lemma \ref{lemma:smoother}]
The proof follows a sequential estimation argument. At each time instance, a new measurement (i.e., channel output) is made available to the decoder which can improve in turn its estimate of the first channel noise instance $\vz_0$. The derivation mostly focuses on writing the channel output as a simple linear function of $\vz_0$, and then we apply known recursive formulas for updating a new estimate and error covariance given a new measurement. We iterate that the derivations hold for general MIMO channels.


Recall that the channel output can be written as
\begin{align}\label{eq:scheme_proof_y_1}
    \vy_n &= \Lambda \vx_n + \vz_n \nn\\
    &\stackrel{(a)}= \Lambda\Gamma \hat{\Sigma}^\dagger(\vhs_{n} - \vhhs_{n}) + H\vhs_{n} + (\vz_n-H\vhs_{n})\nn\\
&= (\Lambda\Gamma\hat{\Sigma}^\dagger + H ) (\vhs_{n}-\vhhs_{n}) + \ve_n + H \vhhs_{n},
\end{align}
where $(a)$ follows from the channel input $\vx_n = \Lambda\Gamma\hat{\Sigma}^\dagger(\vhs_n - \vhhs_n)$. We now relate the channel output $\vy_n$ and $\vz_0$. To this end, the estimation error can be written as the recursion
\begin{align}\label{eq:proof_scheme_error}
    \vhs_{n+1}-\vhhs_{n+1}&= F\vhs_{n} + K_p\ve_{n} - (F\vhhs_n + K_{Y,n}(\vy_n- H\vhhs_n))\nn\\
    &\stackrel{(a)}= F(\vhs_{n}-\vhhs_n) + K_p(\vy_n - H \vhhs_{n} - (\Lambda\Gamma\hat{\Sigma}^\dagger + H ) (\vhs_{n}-\vhhs_{n}) ) - K_{Y,n}(\vy_n- H\vhhs_n)\nn\\
    &= (F-K_p(\Lambda \Gamma\Sigma^\dagger +H))(\vhs_{n} - \vhhs_n) + (K_p-K_{Y,n})(\vy_n  - H \vhhs_{n}) \nn\\
    &\stackrel{(b)}= F_p(\vhs_{n} - \vhhs_n) + (K_p-K_{Y,n})\tilde{\vy}_n \nn\\
    &=  F_p^{n}(\vhs_{1} - \vhhs_1) + \sum_{i=1}^n F_p^{n-i}(K_p-K_{Y,i})\tilde{\vy}_i \nn\\
    &\stackrel{(c)}= F_p^{n}K_p\vz_0 + \mathbf{d}_n,
\end{align}
where $(a)$ follows from \eqref{eq:scheme_proof_y_1}, $(b)$ follows from $F_p \triangleq F-K_p(\Lambda \Gamma\Sigma^\dagger +H)$ and $\tilde{\vy}_n \triangleq \vy_n-H\vhhs_n$, and $(c)$ follows from $\vhs_1 = K_p\vz_0$, $\vhhs_1=0$, and $\mathbf{d}_n \triangleq \sum_{i=1}^n F_p^{n-i}(K_p-K_{Y,i})\tilde{\vy}_i$.

We combine \eqref{eq:scheme_proof_y_1} and \eqref{eq:proof_scheme_error} to write the channel output as
\begin{align}\label{eq:scheme_proof_y_middle}
    \vy_n&=  (\Lambda\Gamma\hat{\Sigma}^\dagger + H)  (F_p^{n-1}K_p\vz_0 + \mathbf{d}_{n-1}) + \ve_n +H\vhhs_n\nn\\
    &= \kappa_n\vz_0 + (\Lambda\Gamma\hat{\Sigma}^\dagger + H) \mathbf{d}_{n-1} + \ve_n +H\vhhs_n,
\end{align}
with $\kappa_n \triangleq (\Lambda\Gamma\hat{\Sigma}^\dagger + H ) F_p^{n-1}K_p$. The estimation model in \eqref{eq:scheme_proof_y_middle} is a sequential estimation problem, but the terms $\mathbf{d}_{n-1}$ and $H\vhhs_n$ on the right-hand side depend on the previous measurement. We proceed to show that these \emph{bias terms} have no effect on the estimation problem and thus can be ignored. Define the transformed measurements (channel outputs)
\begin{align}\label{eq:proof_scheme_o}
 \mathbf{o}_n&\triangleq \tilde{\vy}_n - (\Lambda\Gamma\hat{\Sigma}^\dagger + H ) \mathbf{d}_{n-1}\nn\\
 &= \kappa_n\vz_0 + \ve_n,
\end{align}
in order to obtain a sequential estimation problem (without the bias terms) where the source is $\vz_0$, and at each time we observe $\mathbf{o}_n$ with the measurement noise $\ve_n$. The transformation $\{\tilde{\vy}_i\}_{i=1}^n\to \{\mathbf{o}_i\}_{i=1}^n$ is linear and causal (lower-triangular). Also note that this transformation is invertible since $\mathbf{d}_{n-1}$ is a function of $\tilde{\vy}_1,\dots,\tilde{\vy}_{n-1}$ only. Informally, the invertibility of this transformation shows that the information that can be extracted from the original channel outputs and the transformed channel outputs is the same. Formally, the innovations of both processes are the same, i.e., \begin{align}\label{eq:innovations_equal}
\mathbf{o}_n - \E[\mathbf{o}_n|\mathbf{o}^{n-1}]&= \vy_n - \E[\vy_n|\vy^{n-1}].
\end{align}
We are now ready to present the recursions for the sequential estimation problem in \eqref{eq:proof_scheme_o}. \textcolor{black}{Since the source is the same at all times (i.e., $\vz_0$), we only need a measurement-update formula (e.g., \cite[Lemma $9.3.2$]{kailath_booklinear})} to write it recursively as
\begin{align}\label{eq:proof_scheme_covariance_update}
    &\hat{\vz}_{0|n}= \hat{\vz}_{0|n-1} \nn\\
    &\ + \hat{Z}_{0|n-1} \kappa_n^T \cov(\vy_n |\vy^{n-1})^{-1}(\vy_n - \E[\vy_n|\vy^{n-1}])\nn\\
    &\hat{Z}_{0|n+1} = \hat{Z}_{0|n}\nn\\
    &\ - \hat{Z}_{0|n}\kappa_n^T \cov(\vy_n|\vy^{n-1})^{-1} \kappa_n \hat{Z}_{0|n}
\end{align}
with the initial conditions $\hat{\vz}_{0|0}=0$ and $\hat{Z}_{0|0} = \Psi$. By Lemma \ref{lemma:state-space}, we have $\Psi_{Y,n} = \cov(\vy_n|\vy^{n-1})$ and $\E[\vy_n|\vy^{n-1}] = H\vhhs_n$ so that the recursions simplify to
\begin{align}\label{eq:proof_coding_final_recursions}
       \hat{\vz}_{0|n}&= \hat{\vz}_{0|n-1} + \hat{Z}_{0|n-1}\kappa_n^T  \Psi_{Y,n}^{-1}(\vy_n-H\vhhs_{n})\nn\\
      \hat{Z}_{0|n} &= (I - \hat{Z}_{0|n-1}\kappa_n^T \Psi_{Y,n}^{-1} \kappa_n)\hat{Z}_{0|n-1}.
\end{align}
Furthermore, due to the optimal inputs distribution, the innovation covariance $\Psi_{Y,n}$ converges to its optimal value $\Psi_{Y}^\ast$ (for more details, see the proof of Lemma \ref{lemma:achievable}). Finally, taking a determinant over \eqref{eq:proof_coding_final_recursions}, applying Sylvester's identity, and note that $\Psi_{Y,n}= \kappa_n \hat{Z}_{0|n-1} \kappa_n^T + \Psi$ in \eqref{eq:innovations_equal} gives \eqref{eq:lemma_cov_reduction_det}.
\end{proof}

\section{Conclusions and Future Work}\label{sec:conclusion}
In this paper, we solved the feedback capacity problem of the Gaussian MIMO channel when the noise is generated from a linear dynamical system. The derivation relies on a sequential convex optimization formulation for the finite-block capacity problem using tools from control theory and convex optimization methods. Using the optimization problem convexity along with properties of Riccati recursions convergence, we provided tight lower and upper bounds that resulted a single-letter, computable capacity expression. Additionally, we showed that that the optimization problem induces a time-invariant capacity-achieving inputs distribution that was used to construct an explicit coding scheme for scalar channels.

In a broader perspective, we derived a single-letter formula for the directed information and its main steps can be summarized as follows
\begin{align}\label{eq:conclu_single}
    I(X^n\to Y^N)&= \sum_{i=1}^n I(X^i;Y_i|Y^{i-1}) \nn\\
    &\stackrel{(a)}= \sum_{i=1}^n I(X_i,\hat{S}_i(X^{i-1},Y^{i-1});Y_i|Y^{i-1})\nn\\
    &\stackrel{(b)}= \sum_{i=1}^n I(X_i,\hat{S}_i(X^{i-1},Y^{i-1});Y_i|\hat{\hat{S_i}}(Y^{i-1}))\nn\\
    &\approx n I(X,\hat{S};Y|\hat{\hat{S}}),
\end{align}
where in $(a)$ a channel state $\hat{S}_i = \E[S_i|X^{i-1},Y^{i-1}] = \E[S_i|Z^{i-1}]$ is defined and satisfies the channel Markov chain $(\hat{S}_{i+1},Y_i) - (X_i,\hat{S}_i) - (X^{i-1},Y^{i-1},\hat{S}^{i-1})$, and in $(b)$ $\hat{\hat{S}}_i \triangleq \E[\hat{S}_i|Y^{i-1}]$. Note that the channel state $\hat{S}_i$ can be computed at the encoder since it is a function of $(X^{i-1},Y^{i-1})$. Also, the computation of the asymptotic behaviour at the last step was enabled due to the description of the channel outputs process structure as a hidden-Markov (Lemma \ref{lemma:state-space}).

The above steps are related to computations of the directed information for the discrete-alphabet counterpart of the Gaussian channel, the finite-state channel (FSC). More specifically, for FSCs with state that can be computed at the encoder, the directed information can be written as
\begin{align}\label{eq:DI_1}
    I(X^n\to Y^n)&= \sum_{i=1}^n I(X_i,S_i;Y_i|Y^{i-1}),
\end{align}
and could be expressed with a computable expression in few instances only \cite{Chen05,Ising_channel,Ising_artyom_IT,PeledSabagBEC,Sabag_BIBO_IT,AharoniSabag_RL,PermuterCuffVanRoyWeissman08,trapdoor_generalized}. In \cite{Sabag_UB_IT,sabag_huleihel_TCOM}, it was shown that all these solutions can be unified with the single-letter expression $I(X,S;Y|Q)$, where the channel outputs is a hidden Markov model and $Q$ serves as its hidden state with a finite, graphical structure (called the $Q$-graph). As the conjectured formula structure resembles the one for the Gaussian channel in \eqref{eq:conclu_single}, it should be interesting to investigate whether the techniques developed here apply also for FSCs. In particular, the main step is the formulation of the directed information as a sequential convex optimization problem in order to have an alternative feedback capacity formula that can be \emph{single-letterized}.




Two more research directions are as follows.
\subsubsection{Explicit formulae}
It may be possible to find simple capacity expressions for particular noise processes using the convex optimization in Theorem \ref{th:main}. For instance, the capacity of the ARMA noise of first order can be expressed as a function of the positive root to a quartic equation \cite{Kim10_Feedback_capacity_stationary_Gaussian}. This implies that the two decision variables in Theorem \ref{th:scalar} can be reduced to a single variable. Pursuing such simplifications for ARMA processes of higher order is natural \cite{Butman_conjecture,Butman69,TiernanSchalk_AR_UB,Han_GaussianFeedback}.
\subsubsection{Scheme for MIMO channels}
In Section \ref{sec:scheme}, we presented an explicit scheme for scalar channel that trivially extends to MIMO channels that can be decomposed to parallel scalar  channels. However, an explicit scheme for non-trivial MIMO channels remains open. A conjectured scheme was described in Section \ref{sec:scheme}, and a refinement of the spectral analysis in Lemma \ref{lemma:smoother} should prove the scheme optimality.

\bibliography{ref}
\bibliographystyle{IEEEtran}

\end{document}